\makeatletter
\disable@package@load{fixltx2e}{}
\makeatother
\documentclass[final,5p,times,twocolumn]{elsarticle}
\usepackage[T1]{fontenc}

\DeclareFontShape{T1}{txr}{m}{scit}{<-> ssub * txr/m/sc}{}
\DeclareFontShape{T1}{txr}{bx}{scit}{<-> ssub * txr/bx/sc}{}
\usepackage{amssymb}
\usepackage{amsmath}
\usepackage{amsthm}
\usepackage{booktabs}
\usepackage{array}
\usepackage{tabularx}
\usepackage{xcolor}
\usepackage{graphicx}
\usepackage{xurl}
\usepackage[hypertexnames=false]{hyperref}

\hypersetup{hidelinks}
\hypersetup{hidelinks}
\usepackage{enumitem}
\usepackage{algorithm}
\usepackage{algpseudocode}

\usepackage{float}
\usepackage{dblfloatfix}

\newcolumntype{L}[1]{>{\raggedright\arraybackslash}p{#1}}
\newcolumntype{C}[1]{>{\centering\arraybackslash}p{#1}}
\newcolumntype{Y}{>{\raggedright\arraybackslash}X}

\newtheorem{proposition}{Proposition}[section]
\newtheorem{corollary}{Corollary}[section]
\theoremstyle{definition}
\newtheorem{definition}{Definition}[section]
\colorlet{blue}{black}

\begin{document}

\newdimen\arxivSavedPdfHOrigin
\arxivSavedPdfHOrigin=\pdfhorigin
\AddToHookNext{shipout/before}{\global\advance\pdfhorigin by 0.22in}
\AddToHookNext{shipout/after}{\global\pdfhorigin=\arxivSavedPdfHOrigin}

\begin{frontmatter}

\title{\textcolor{blue}{Policy-Controlled Offline-Verifiable Accountability for Cross-Organization Agent Messaging: A Preserved Evidence-Bundle Approach}}
\author[a,b]{Adil Alshammari\corref{cor1}}
\ead{aha388@nau.edu}

\author[a,c]{Hayretdin Bahsi}
\ead{Hayretdin.Bahsi@nau.edu}

\cortext[cor1]{Corresponding author.}

\affiliation[a]{
    organization={School of Informatics, Computing, and Cyber Systems,
                  Northern Arizona University},
    city={Flagstaff},
    state={Arizona},
    country={United States}
}

\affiliation[b]{
    organization={College of Computer and Information Sciences,
                  Majmaah University},
    city={Majmaah},
    postcode={11952},
    country={Saudi Arabia}
}

\affiliation[c]{
    organization={School of Information Technologies,
                  Tallinn University of Technology},
    city={Tallinn},
    country={Estonia}
}

\begin{abstract}
Cross-organization agent workflows require preserved evidence that remains independently verifiable during later audit or dispute review. Such workflows may involve multiple organizations, delegated actions, policy-relevant events, and disputed accountability claims. This is difficult when live systems are unavailable, controlled by one party, or not trusted by all participants. Existing mechanisms provide useful pieces, including authenticated logging, delegation semantics, signed checkpoints, and consistency checks. What remains missing is a verifier-centered event-level bundle for checking evidence sufficiency offline under an explicit policy.

We propose a preserved evidence-bundle model and a policy-controlled offline verifier for agent-to-agent workflow events. Each bundle preserves policy-required evidence, including sender authentication, authenticated log commitment, witness-backed checkpoint evidence, append-only continuity, delegation-aware authorization evidence, and, when issued and required by policy, explicit receiver-signed receipt evidence. 

The verifier accepts \textcolor{blue}{a bundle} only \textcolor{blue}{when all mandatory and policy-required checks succeed}, giving a later dispute reviewer an offline basis for assessing evidence sufficiency. It does not infer delivery or receipt from transport behavior or log inclusion alone.

\textcolor{blue}{In a prototype evaluation, the event-type study used 300 synthetic four-event healthcare workflow cases and 1200 valid preserved bundles. The policy-profile comparison observed its largest median latency reduction relative to the fully enabled profile when checkpoint-context validation was disabled. All 17 tested cases in the original negative-evidence diagnostic suite were rejected under the fully enabled profile. Additional functional studies examined context and signed-reference binding. These results support prototype feasibility for the evaluated workloads.}
\end{abstract}

\end{frontmatter}

\section{Introduction}
\label{sec:introduction}
In cross-organization agent workflows, structured messages carry requests, results, and policy-relevant actions across administrative boundaries. Therefore, a later dispute may involve more than whether a message was sent. It may also involve authorship, authenticated logging, policy compliance, delegated authority, or whether receipt is supported by explicit receiver-signed evidence.

\textcolor{blue}{Later accountability review requires independently verifiable evidence because operational records may be controlled by one party, altered, or unavailable, and live services may be inaccessible. Studies of multi-agent systems and AI-agent relationships discuss contestability, accountability, and credible agent identity across organizational and service boundaries \citep{nguyen_position_2026,lange_we_2025,lin_binding_2025}.}

This evidence problem becomes concrete in sectors such as healthcare. Prior work in this domain has examined privacy-preserving medical data protection and controlled access \citep{jakhar_blockchain-based_2024,alshammari_efficient_2025}. Agentic AI and AI-agent systems are increasingly being studied for healthcare workflows and clinical decision support \citep{collaco_role_2026,zhao_ai_2026,zheng_large_2025}. A single workflow may involve several organizations, delegated actions, and later audit or dispute review. In such workflows, a later reviewer may need evidence that can be checked without relying only on the systems that originally produced or stored the event. More broadly, recent surveys describe autonomous, goal-oriented, and workflow-oriented agent systems beyond the healthcare example \citep{acharya_agentic_2025,yu_survey_2025}.  

Industry reporting suggests that this is not only a theoretical concern. For example, a Cloud Security Alliance report states that 53\% of organizations observed AI agents exceeding intended permissions at least occasionally, and 47\% reported a security incident involving AI-agent behavior
\citep{cloud_security_alliance_more_2026}. Surveys and assurance work on agentic large language model (LLM) deployments raise related concerns about permissions, delegated tool use, and third-party service boundaries
\citep{li_security_2025,stamnes_karlsen_securing_2026}. For this paper, the implication is narrower: the verifier should be able to check preserved evidence directly, instead of relying on platform-specific logs, live service access, or informal operational assurances.

\textcolor{blue}{Agent workflows motivate this work because an autonomous component may select subsequent workflow actions, invoke external services, or act on behalf of another entity under delegated authority. For example, an agent acting for one organization may request a service from another organization. A later dispute may concern who issued the request and whether the requester had the required authority to act on behalf of that organization. These questions motivate preserving sender-authentication evidence and, when required by policy, delegation evidence with the workflow event. We therefore treat agent-to-agent messages as observable workflow events and evaluate the preserved evidence under the selected policy, without reconstructing the agent's internal decision process.}

\textcolor{blue}{The verification model is intentionally more general than the agent setting. It operates on structured signed messages and preserved cryptographic
evidence. It also applies to conventional distributed services and
cross-organization workflows that preserve the artifacts required
by the selected policy. Agent communication is therefore an important application of the model, whose verification rules do not depend on an LLM or a particular agent architecture. Acceptance remains limited to the evidence-level claims supported by the selected policy; it does not establish the semantic correctness of an autonomous decision.}

Prior work gives useful building blocks, but most of them are studied in separate settings. Transparency-logging and key-transparency work provide authenticated commitments, signed log states, and consistency checks
\citep{dirksen_logpicker_2021,malvai_parakeet_2023}. Delegation and auditable-authorization work addresses authority transfer, authorization provenance, and policy-governed access in multi-party environments
\citep{ferretti_verifiable_2021,li_damfsd_2024}. Formal workflow-delegation studies also show that delegated authority can interact with
separation-of-duty constraints and workflow satisfiability \citep{yang_delegation_2024}. Formally verified accountability protocols define non-repudiation and evidence-driven verdict semantics for message exchange \citep{falanji_momep_2026}.

\textcolor{blue}{Our earlier work established offline verification of message authorship and log inclusion at a signed Merkle snapshot \citep{alshammari_authenticated_2026}. This paper extends that foundation to policy-relative workflow-event accountability using additional checkpoint, delegation, and prerequisite evidence. Section~\ref{sec:related} provides the contribution comparison.}

This leads to the following research question: \textbf{How can message-level evidence for agent-to-agent workflows be extended into policy-controlled workflow-event evidence that enables an offline verifier to decide whether a preserved bundle satisfies the selected policy profile and its required evidence checks?}

This question is addressed through four design objectives and claim boundaries. First, evidence is preserved per workflow event because later disputes usually concern a specific workflow action, not only the overall system state. Second, verification is performed offline because live services may be unavailable, controlled by one party, or not trusted during later review. Third, receipt requires explicit receiver-signed evidence because receipt cannot be inferred from transport behavior, storage, or log inclusion. Fourth, acceptance remains bounded by the selected policy so that the verifier does not treat a valid signature, log commitment, or capability chain as proof of delivery, endpoint correctness, semantic correctness, or workflow completion.

The novelty of this paper is the policy-controlled, workflow-event-level verification of evidence sufficiency. For a specific agent-to-agent workflow event, the offline verifier checks whether the preserved bundle satisfies the evidence checks required by the selected policy profile. The key question is not only whether individual artifacts are valid, but whether the preserved evidence is sufficient to support the claim under that policy. Acceptance is interpreted according to these claim boundaries, and the evaluation is claim-matched instead of a general throughput benchmark.

The paper makes three contributions:
\begin{enumerate}[label=\textbf{\arabic*)}, leftmargin=1.5em]
\item \textbf{Preserved evidence-bundle model.} We introduce a verifier-centered evidence-bundle model for agent-to-agent workflow events in cross-organization workflows. The model preserves sender, log, receipt, checkpoint, continuity, and delegation evidence as independently checkable evidence classes.

\item \textbf{Policy-controlled workflow-event verifier.} We define a policy-controlled offline verifier for workflow-event evidence. The selected policy profile determines both which evidence predicates must succeed before the verifier returns \textsc{Accept} and which evidence-level claims are supported by the acceptance decision. Therefore, stronger claims require the corresponding preserved evidence, such as receipt, witness-backed checkpoint, continuity, checkpoint-context, or delegation evidence.

\item \textbf{Claim-matched evaluation.}
\textcolor{blue}{We implement a prototype and evaluate
verification performance, diagnostic rejection behavior,
context and reference binding, input handling, and
sensitivity to evidence and workflow size.}
\end{enumerate}

The paper is structured as follows: Section~\ref{sec:related} reviews related
work. Section~\ref{sec:systemmodel} presents the system model, assumptions,
threat model, and evidence-control boundaries. Section~\ref{sec:design}
describes the system design and evidence lifecycle. Section~\ref{sec:offline} presents the policy-controlled offline verification procedure. Section~\ref{sec:semantics} defines the verifier acceptance semantics and bounded claims. Section~\ref{sec:security_analysis} provides an evidence-level security assessment. Section~\ref{sec:evaluation} reports the evaluation results. Section~\ref{sec:discussion} discusses implications and limitations.
Finally, Section~\ref{sec:conclusion} concludes the paper and presents future work.

\section{Related Work}
\label{sec:related}

\begin{table*}[!t]
\centering
\begingroup
\color{blue}
\caption{Comparison of related work by verifier-relevant evidence support.}
\label{tab:rw_jisa}
\scriptsize
\renewcommand{\arraystretch}{1.12}
\setlength{\tabcolsep}{3pt}

\begin{tabularx}{\textwidth}
{@{}L{0.10\textwidth}
C{0.035\textwidth}
L{0.095\textwidth}
Y
*{4}{C{0.085\textwidth}}@{}}
\toprule
\textbf{Reference} &
\textbf{Year} &
\textbf{Context} &
\shortstack[l]{\textbf{Evidence and}\\\textbf{verification target}} &
\shortstack{\textbf{Signed}\\\textbf{checkpoint}\\\textbf{(log)}} &
\shortstack{\textbf{Inclusion}\\\textbf{proof}\\\textbf{(log)}} &
\shortstack{\textbf{Witness /}\\\textbf{consistency}} &
\shortstack{\textbf{Delegation}\\\textbf{evidence}} \\
\midrule

Dirksen et al.~\citep{dirksen_logpicker_2021} &
2021 &
Web PKI &
Certificate logging promises and client proofs;
inter-log inclusion and consistency audits. &
$\checkmark$ &
$\checkmark$ &
$\checkmark$ &
$\times$ \\

\addlinespace[2pt]
Malvai et al.~\citep{malvai_parakeet_2023} &
2023 &
E2EE messaging &
Signed key-directory responses;
key membership against witness-certified state. &
$\triangle^{a}$ &
$\triangle^{a}$ &
$\checkmark$ &
$\times$ \\

\addlinespace[2pt]
Li et al.~\citep{li_damfsd_2024} &
2024 &
Medical data &
Delegated-permission credentials;
verification of root and direct authorization sources. &
$\times$ &
$\times$ &
$\times$ &
$\checkmark$ \\

\addlinespace[2pt]
Ferretti et al.~\citep{ferretti_verifiable_2021} &
2021 &
IIoT &
Authorization-grant packages and audit evidence;
local grant verification and log membership. &
$\checkmark$ &
$\checkmark$ &
$\times$ &
$\checkmark$ \\

\addlinespace[2pt]
\textbf{This work} &
\textbf{2026} &
\textbf{Agent messaging} &
\textbf{Preserved workflow-event bundles and associated offline
inputs; policy-relative evidence sufficiency and bounded claim
support.} &
$\checkmark$ &
$\checkmark$ &
$\checkmark^{*}$ &
$\checkmark$ \\

\bottomrule
\end{tabularx}

\par\smallskip
\parbox{\textwidth}{\footnotesize
\textit{Note:}
$\checkmark$ = meets the stated column criterion;
$\triangle$ = a specified related mechanism does not meet one
column condition, as explained below;
$\times$ = no qualifying mechanism is specified in the inspected
protocol.
Ratings credit the specified issuance, query, or audit procedures;
they do not imply identical packaging, deployment assumptions,
or security guarantees.

$^{a}$~Parakeet's signed, witness-certified state and membership
proofs concern a key directory, rather than the append-only log
checkpoints and log-record inclusion defined here
(Sections III-B--D and V-A).
This target difference does not indicate weaker security.

$^{*}$~Our prototype verifies checkpoint co-signatures using distinct
log and witness keys under shared administration.
Administrative independence remains a deployment trust assumption.
WCL and prerequisite bundles are separate offline inputs.
}
\endgroup
\end{table*}

\textcolor{blue}{The closest related work falls into four lines: transparency logging, key transparency, witness-assisted consistency, and delegation-aware authorization. Transparency logging provides append-only log commitments, signed log states, and inclusion or consistency evidence \citep{laurie_certificate_2021,dahlberg_verifiable_2018,dirksen_logpicker_2021}. Key-transparency systems provide verifiable support for identity-to-key bindings and key-directory consistency \citep{malvai_parakeet_2023}. Witness-assisted consistency uses additional observation or consistency evidence to reduce reliance on a single log view \citep{dirksen_logpicker_2021,malvai_parakeet_2023}. Delegation-aware authorization preserves evidence about authority transfer or policy-governed action in multi-party settings \citep{ferretti_verifiable_2021,li_damfsd_2024}. In this paper, delegation means that one actor performs a workflow action on behalf of another actor under preserved authority evidence, such as a capability chain or a policy-bound authorization artifact.}

\textcolor{blue}{On the logging side, the main cryptographic basis is Merkle/hash-tree commitment and certificate-transparency-style append-only logging \citep{merkle_digital_1988,laurie_certificate_2021,dahlberg_verifiable_2018}. These mechanisms underlie the inclusion proofs, signed log-state checkpoints, and consistency proofs used in append-only logs.}

\textcolor{blue}{Related systems use authenticated data structures
for efficient log monitoring, workload-specific log construction,
and verifiable ledger storage
\citep{hu_merkle2_2021,fang_legolog_2025,yue_glassdb_2023-1}.}

At the deployment level, work on multi-ownership digital service chains shows how security responsibilities and security-relevant information can span multiple providers and administrative domains
\citep{repetto_cybersecurity_2026}. This supports our focus on verifier-centered artifacts for offline assessment.

\textcolor{blue}{
Table~\ref{tab:rw_jisa} distinguishes evidence targets from four
independently assessed verification mechanisms.
\emph{Evidence and verification target} summarizes the supplied
artifacts and the object or relationship checked in the specified
procedures.
All compared designs specify local cryptographic checks of supplied
artifacts using available keys and trust inputs, without live
service requests for those checks; issuance, retrieval, tracking,
and revocation may involve services.
\emph{Signed checkpoint (log)} means a signed commitment to a
snapshot of a specified append-only log; continuity between
snapshots is a separate check.
\emph{Inclusion proof (log)} means a proof that an identified record
belongs to that authenticated log snapshot.
\emph{Witness/consistency} means a signed endorsement of checkpoint
or state material by an additional observer, or an explicit
consistency check between authenticated views.
Ordinary monitoring, digest recomputation, or blockchain consensus
alone does not satisfy this criterion.
\emph{Delegation evidence} means cryptographically verifiable
evidence of transferred authority identifying the relevant actor
and scope.
The same criteria apply to all rows, including our prototype.
}

Formal work on non-repudiation, evidence-driven verdicts, and accountability shows how preserved evidence can support later accountability decisions \citep{falanji_momep_2026,kusters_accountability_2010}. This motivates our verifier model, but our scope is narrower. The verifier checks the evidence preserved for one workflow event under the selected policy.

\textcolor{blue}{Our earlier work established offline message-level checks for payload-hash binding, sender authentication, Merkle inclusion, and signed-checkpoint authentication \citep{alshammari_authenticated_2026}. The present paper retains these checks and adds policy-relative acceptance and claim support for preserved workflow events.}

\begin{table*}[!t]
\centering
\color{blue}
\caption{Extensions beyond the earlier message-level model.}
\label{tab:prior_work_delta}
\small
\renewcommand{\arraystretch}{1.12}
\setlength{\tabcolsep}{4pt}

\begin{tabularx}{\textwidth}{@{}L{0.16\textwidth} L{0.38\textwidth} L{0.25\textwidth} Y@{}}
\toprule
\textbf{Extension} &
\textbf{Evidence and verification check} &
\textbf{Supported property} &
\textbf{Evaluation section} \\
\midrule

Policy-relative review &
Profile and event bundle; check mandatory and
profile-required predicates. &
Evidence sufficiency under the selected profile. &
\ref{subsec:policy_tradeoff},
\ref{subsec:tamper_eval} \\

Witness endorsement &
Checkpoint co-signature; verify the witness signature. &
Checkpoint endorsement by the identified witness key. &
\ref{subsec:tamper_eval} \\

Log continuity &
Signed checkpoints and extension proof; verify
signatures and append-only extension. &
Authenticated extension between the supplied checkpoints. &
\ref{subsec:policy_tradeoff},
\ref{subsec:tamper_eval} \\

Checkpoint context &
WCL and target checkpoint; verify journal records
and match checkpoint signing bytes. &
Integrity of the supplied WCL and presence
of the target checkpoint signing bytes. &
\ref{subsec:policy_tradeoff},
\ref{subsec:wcl_stage_profile},
\ref{subsec:tamper_eval} \\

Delegation and prerequisites &
Capability chain and prior-event bundle; check
authority, message binding, and prerequisites. &
Capability-chain validity and prerequisite agreement. &
\ref{subsec:policy_tradeoff},
\ref{subsec:event_eval},
\ref{subsec:tamper_eval} \\
\textcolor{blue}{Receiver admission} &
\textcolor{blue}{Receiver-signed admission receipt; check key
authorization, signature, and binding to the message occurrence,
receiver, admission policy and version, and admission state.} &
\textcolor{blue}{Authenticity and context binding of the
receiver's admission assertion.} &
\textcolor{blue}{\ref{subsec:receipt_eval}} \\
\bottomrule
\end{tabularx}

\par\smallskip
\begin{minipage}{0.98\textwidth}
\scriptsize
\emph{Notes:}
The original WCL-omission diagnostic (C14) tests
unavailability of the required file.
Section~\ref{subsec:wcl_stage_profile} reports separate
stage profiling and focused WCL controls.

\textcolor{blue}{Receiver admission receipts follow
Definition~\ref{def:explicit-receipt-predicate} and are implemented
and functionally evaluated in a separate revision candidate
(Section~\ref{subsec:receipt_eval}).}

\textcolor{blue}{The signed-reference correction receives focused functional
validation in Section~\ref{subsec:tamper_eval}. The revised E1 and E2
measurements use a verifier containing that correction; E3 retains its
original diagnostic implementation.}

\end{minipage}
\end{table*}

\textcolor{blue}{Experiments E1--E3 evaluate policy profiles, event-specific evidence requirements, and targeted rejection behavior, respectively (Section~\ref{sec:evaluation}).}

\textcolor{blue}{The verifier checks a signed \texttt{prior\_evidence\_ref}
against the validated prerequisite message and logged occurrence, binding
the claim to the particular evidence committed by its reference.
Experiment E7 (Section~\ref{subsec:tamper_eval}) provides focused functional
validation of this binding. Its verifier calls are separate from the
revised E1 and E2 timing samples.}

\textcolor{blue}{Overall, the related work provides important mechanisms for authenticated state, consistency evidence, authorization provenance, policy-governed access, capability-based access control, accountable credentials, and privacy-preserving access control \citep{dirksen_logpicker_2021,malvai_parakeet_2023,ferretti_verifiable_2021,li_damfsd_2024,zhu_group-capability-based_2025,tegane_extended_2023,cheng_s-cred_2024,hu_towards_2024}. The present contribution is the explicit composition of evidence predicates for reviewing a preserved signed workflow event under a supplied policy profile. Acceptance requires every mandatory and profile-required predicate to succeed. In the model, a claim is supported only when its required predicates are included in the profile and succeed. Table~\ref{tab:prior_work_delta} distinguishes the additions beyond our earlier message-level baseline~\citep{alshammari_authenticated_2026} and maps them to their artifacts, verification checks, supported properties, and evaluation.}

\section{System Model, Assumptions, and Threat Model}
\label{sec:systemmodel}

\subsection{Deployment Context and Accountability Scope}
\label{subsec:context}
The model covers cross-organization workflows in which autonomous or semi-autonomous software agents exchange structured messages across administrative boundaries. A deployment may involve several institutions and several steps, with one organization initiating a request and other parties performing delegated actions. A verifier may later review the
preserved evidence without contacting the original runtime environment. The examples are illustrative and do not limit the domain. The framework applies when accountability depends on preserved cryptographic evidence. It does not assume continued access to online services.

The unit of analysis is the message event, not generalized application state. Each accountability decision is anchored to a preserved evidence bundle for a specific message event and, when required by policy, its checkpoint and delegation context. We do not reconstruct arbitrary execution state. The verifier instead checks bounded claims about message authorship, authenticated logging, explicit receipt, append-only continuity, and delegation-aware authorization.

In this model, each evidence class supports a different bounded claim.
Sender-authenticated inputs support exact-message binding. Inclusion and
signed-checkpoint evidence support authenticated log commitment. Explicit
signed receipts support a receipt claim only when preserved and verified.
Witness evidence supports checkpoint validation. Append-only extension proof
(AOEP) evidence supports continuity between authenticated checkpoints. Stronger checkpoint reasoning requires the corresponding checkpoint evidence. Delegation evidence supports authorization-related claims under policy.

\subsection{Trust Assumptions, Adversary Scope, and Claim Boundaries}
\label{subsec:assumptions}

\textcolor{blue}{The framework assumes standard security for the digital signature and hash functions used in the implementation. For deployment, participating organizations must authorize key issuance after establishing the key holder's identity and role, and retain authenticated, dated identity-to-key bindings. Each rotation must use a new key identifier while preserving the previous public key, its validity interval, and an authenticated transition record. Revocation records must identify the issuing authority, affected key, status, issuance time, effective time, and any established compromise interval.}

\textcolor{blue}{Before offline verification, the authorized reviewer must obtain approved trust-anchor and public-key fingerprints through a trusted channel, retain them in the offline audit context, and compare the supplied registry with the approved archived snapshot. \textcolor{blue}{This approval must cover each verification key's
permitted signing role and each capability issuer's authority
for the supplied grants.} The audit context must preserve the public keys, registry versions, rotation and revocation records, applicable policy version, and accepted time basis. Historical key admissibility is assessed using these records and their effective times. A signed checkpoint creation time is a reference supplied by its signer, rather than independent evidence of the event time.}

\textcolor{blue}{The prototype preserves a public-key registry snapshot and resolves keys locally for signature verification. It does not automate key issuance, rotation, revocation checking, or trust-anchor pinning. The authenticity, completeness, and historical suitability of the supplied trust material remain deployment trust assumptions.}

The transport path is the channel used to transfer messages and evidence artifacts between parties. It may use ordinary operational protections, but it is not trusted as an accountability source: messages or artifacts may be delayed, replayed, omitted, reordered, modified, or replaced before offline verification. Therefore, transport success is treated as operational information only, not as an accountability signal. The offline verifier detects tampering, omission, substitution, rebinding, or missing policy-required evidence only when these affect preserved artifacts checked under the selected policy profile. The verifier detects runtime reordering only when it causes a
policy-required evidence check to fail. It does not reconstruct
the full transport history.

We model the adversary primarily as an evidence manipulator, not as an attacker with unrestricted runtime control. The adversary may present incomplete bundles, inconsistent checkpoint material, invalid witness-related evidence, or forged capability and workflow references in an attempt to induce an incorrect acceptance decision.

The adversary model does not assume malicious replacement of the offline verifier's trust anchors. It also does not claim protection against compromised signing keys or collusion among all evidence-producing parties. Such cases are treated as trust-anchor, key-management, or deployment failures rather than evidence-verification failures.

The verifier is assumed to execute correctly once provisioned with the required trust anchors, public verification keys, policy inputs, and policy-permitted offline references. It is limited to preserved artifacts and those local inputs. It does not query runtime services, treat transport success as evidence, perform online freshness checks, or reconcile live system views during verification. \textcolor{blue}{Policy selection is a trusted input to offline review. We assume that the participating organizations establish the policy applicable to the claim and event being reviewed and retain its definition, identifier, version, and association with that event in the agreed audit context. An authorized dispute reviewer supplies the corresponding verification profile. Policy agreement, preservation of the applicable historical version, and authorization of profile selection are deployment trust assumptions.}

\textcolor{blue}{The framework defines a witness as a checkpoint observation service, with organizational independence from the logging service and receiver as a deployment trust assumption. The local prototype uses distinct log and witness signing keys, with witness signing performed inside the receiving/logging process under shared administration.}

\textcolor{blue}{In the prototype, the witness observes and
co-signs checkpoint material.}
It does not receive plaintext payloads, issue receipts, infer receipts,
or decide whether a workflow action is correct.
\textcolor{blue}{Witness evidence supports endorsement of the specified
checkpoint. The offline verifier checks append-only continuity
separately using authenticated checkpoints and an AOEP.}

\textcolor{blue}{When checkpoint-context validation is enabled, the verifier may use preserved material from the witnessed checkpoint log (WCL) as context for checkpoint validation.}

These claim boundaries are necessary because verifier acceptance is limited to claims justified by the preserved evidence and the selected policy profile. The verifier supports an authenticated logged-message claim when sender-authentication and authenticated logging evidence validate. It supports a receipt claim only when an explicit receiver-signed receipt object is preserved and successfully verified. Receipt cannot be claimed from transport acknowledgments, HTTP status codes, forwarding, storage, log inclusion, or execution side effects. Accordingly, verifier acceptance does not establish delivery, endpoint correctness, or semantic correctness of the workflow action. Runtime LLM threats are also outside the verifier's evidence-level scope, including prompt injection, malicious planning, and unsafe tool use.

Table~\ref{tab:boundaries} summarizes the evidence-control boundaries considered in the threat model. An evidence-control boundary is a point in the evidence lifecycle where evidence is generated, transferred, preserved, or evaluated. These boundaries are not trusted zones. They indicate where artifacts may be modified, omitted, substituted, replayed, or rebound before offline verification.

\begin{table}[t]
\centering
\caption{Evidence-control boundaries and verifier implications.}
\label{tab:boundaries}
\footnotesize
\renewcommand{\arraystretch}{1.06}
\setlength{\tabcolsep}{3pt}
\begin{tabularx}{\columnwidth}{@{}L{0.38\columnwidth}X@{}}
\toprule
\textbf{Boundary} & \textbf{Verifier implication} \\
\midrule

Agent-side domains &
Verify sender identity and policy-bound evidence. 
\\

Untrusted transport &
Do not treat transport success as evidence. \\

Evidence-producing services &
Check policy-required evidence. \\

Preserved evidence stores &
Check for omission, substitution, or tampering. \\

Local verifier trust base &
Use local policy, keys, and trust anchors. \\

External action boundary &
Do not infer endpoint correctness or execution. \\

\bottomrule
\end{tabularx}
\end{table}

\subsection{Main Notation}
\label{subsec:notation}
Table~\ref{tab:notation_terms} gives the main notation used in the system description, algorithms, and acceptance semantics. Several paper-specific terms are also used consistently in the sections that follow. An evidence bundle denotes the preserved set of artifacts used by the offline verifier to evaluate one message event. A policy profile specifies which verification predicates are required for acceptance. Authenticated log commitment denotes evidence that a message commitment is included in a signed checkpoint of the append-only log. An explicit signed receipt denotes receiver-signed evidence that supports a receipt claim only when such a receipt object is present and valid.

\begin{table}
\centering
\caption{Main notation.}
\label{tab:notation_terms}
\footnotesize
\renewcommand{\arraystretch}{1.03}
\begin{tabularx}{\columnwidth}{p{0.12\columnwidth} X}
\toprule
\textbf{Notation} & \textbf{Description} \\
\midrule
$m$ & Canonicalized message representation. \\
$h(m)$ & Message commitment derived from $m$. \\
$H(\cdot)$ & Hash or commitment function. \\
$\sigma_S$ & Sender signature over the canonical message. \\
$\sigma_R$ & Explicit receiver-signed receipt. \\
$\sigma_L$ & Log signature over signed-checkpoint fields. \\
$\pi_{\mathrm{incl}}$ & Inclusion proof for the message commitment. \\
$\pi_{\mathrm{ext}}$ & Append-only extension proof (AOEP) between checkpoints. \\
$\mathrm{WCL}$ & \textcolor{blue}{Witnessed checkpoint log used for checkpoint-context validation.} \\
$\mathrm{STH}$ & Signed tree head comprising $r$, $t$, $\tau$, and $\sigma_L$. \\
$r$ & Checkpoint root, as used in $\mathrm{STH}$. \\
$t$ & Authenticated \textcolor{blue}{Merkle-tree size},
as used in $\mathrm{STH}$. \\
$\tau$ & Checkpoint context, e.g., timestamp, as used in $\mathrm{STH}$. \\
$\mathcal{C}_{\mathrm{cap}}$ & Capability chain used as delegation evidence. \\
$P$ & Selected policy profile. \\
$CP$ & Set of verification predicates required under policy $P$. \\
$\Phi$ & Universe of verifier predicates. \\
$\varphi$ & Generic verifier predicate. \\
$\mathcal{B}$ & Preserved evidence bundle. \\
\bottomrule
\end{tabularx}
\end{table}

Append-only continuity means that a later authenticated checkpoint
validly extends an earlier one. We use $\pi_{\mathrm{ext}}$ for the
append-only extension proof (AOEP) that validates this continuity.
Checkpoint-context evidence \textcolor{blue}{comprises the supplied
WCL and the target authenticated checkpoint; verification checks
journal integrity and membership of that checkpoint.}
\textcolor{blue}{A valid witness co-signature authenticates endorsement
of the specified checkpoint by the identified witness key.}
Delegation-aware authorization evidence is preserved authority
evidence showing that one actor was authorized to perform a workflow
action on behalf of another actor, for example through a capability
chain or a policy-bound authorization artifact.

\subsection{Evidence Support and Claim Limits}
\label{subsec:supports_summary}
Table~\ref{tab:boundaries} identifies where evidence may be generated, transferred, preserved, or evaluated in the evidence lifecycle. In contrast, Table~\ref{tab:supports} summarizes the bounded claim support of each preserved artifact: what the artifact supports and what it does not establish on its own.

\begin{table}[H]
\centering
\caption{Evidence artifacts and claim limits.}
\label{tab:supports}
\footnotesize
\renewcommand{\arraystretch}{1.00}
\setlength{\tabcolsep}{2pt}
\begin{tabularx}{\columnwidth}{@{}L{0.25\columnwidth} L{0.31\columnwidth} Y@{}}
\toprule
\textbf{Artifact} & \textbf{Supports} & \textbf{Does not establish} \\
\midrule

Sender signature &
Message authorship &
Logging, receipt, or delivery \\

Log evidence &
Log commitment &
Receipt, delivery, or endpoint correctness \\

Signed receipt &
Receipt claim &
Delivery or downstream execution \\

Witness evidence &
Checkpoint support &
Receipt, delivery, or real-time equivocation prevention \\

AOEP evidence &
Append-only extension &
Global consistency or collusion resistance \\

Checkpoint context &
\textcolor{blue}{Checkpoint-context validation} &
Continuity, receipt, or endpoint correctness \\

Capability chain &
Delegated authorization &
Logging, receipt, or workflow completion \\

\bottomrule
\end{tabularx}
\end{table}
\section{System Design and Evidence Lifecycle}
\label{sec:design}

\subsection{Architecture Overview}
\label{subsec:architecture}

\begin{figure*}[!t]
    \centering
\includegraphics[width=0.85\textwidth]{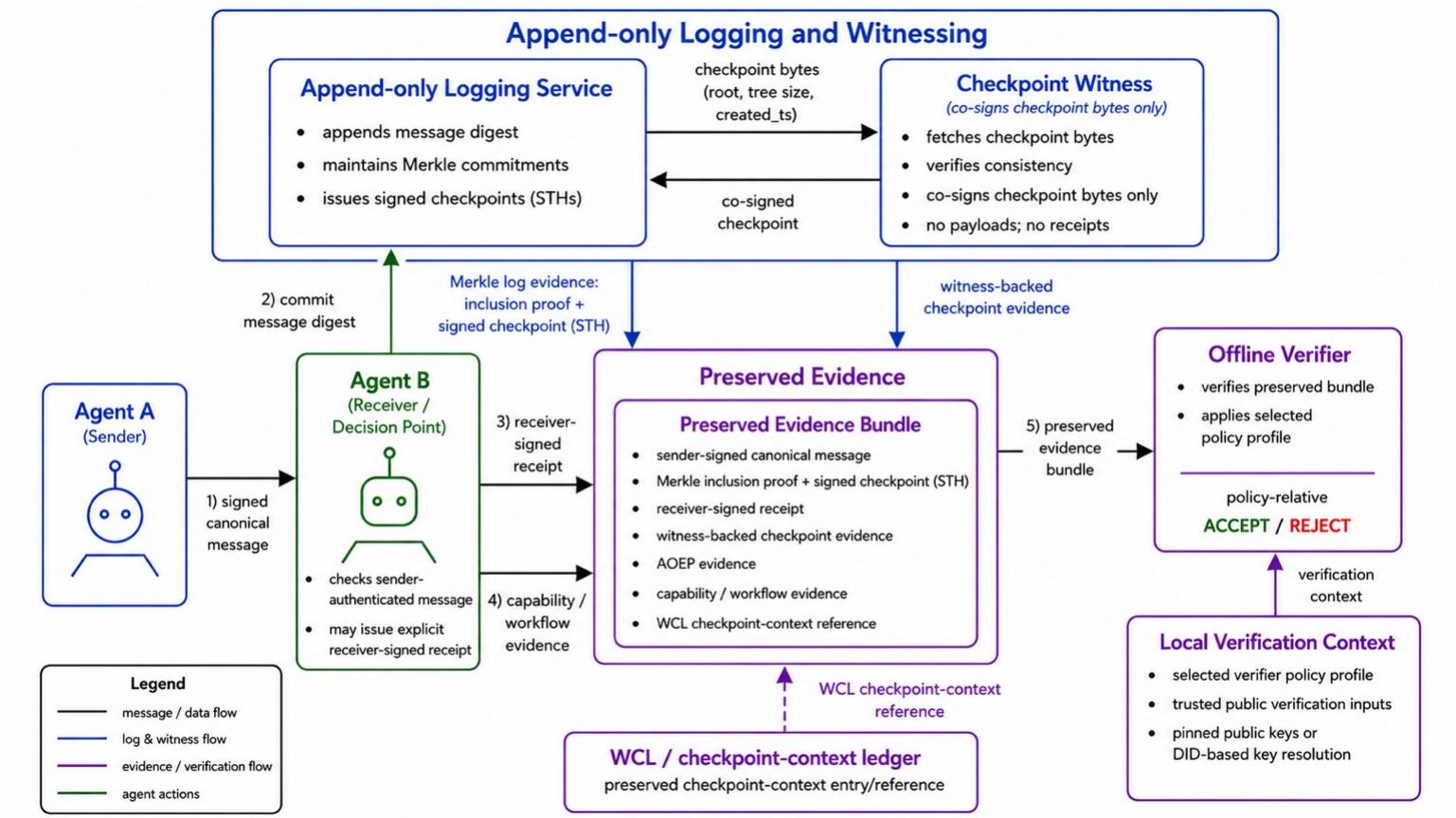}
\caption{Sender-authenticated messages are validated, logged, augmented with policy-required checkpoint evidence, and preserved as evidence bundles for offline verification.}
\label{fig:architecture}
\end{figure*}

Figure~\ref{fig:architecture} shows the high-level architecture and evidence flows of the proposed framework. The design uses five logical roles: a sender agent, a receiver-side decision point, an append-only logging service, a checkpoint-witnessing role, and an offline verifier. 

\textcolor{blue}{The policy controls whether bundle witness evidence
is mandatory; any supplied bundle witness signature is verified.}
We use the term receiver-side decision point because the receiving organization may implement this role as an agent, gateway, policy service, or workflow controller. The required property is the ability to validate the incoming message and attach the evidence required by policy, not a specific software architecture.

The architecture is verifier-centered. Live message processing produces preserved evidence. Later dispute resolution evaluates that evidence offline. This separation allows the verifier to assess the preserved bundle without live access to the sender, receiver-side decision point, logging service, or checkpoint-witnessing role.

\subsection{Evidence Issuance Path}
\label{subsec:issuance}
\begin{figure*}[!t]
  \centering
  \includegraphics[width=0.85\textwidth]{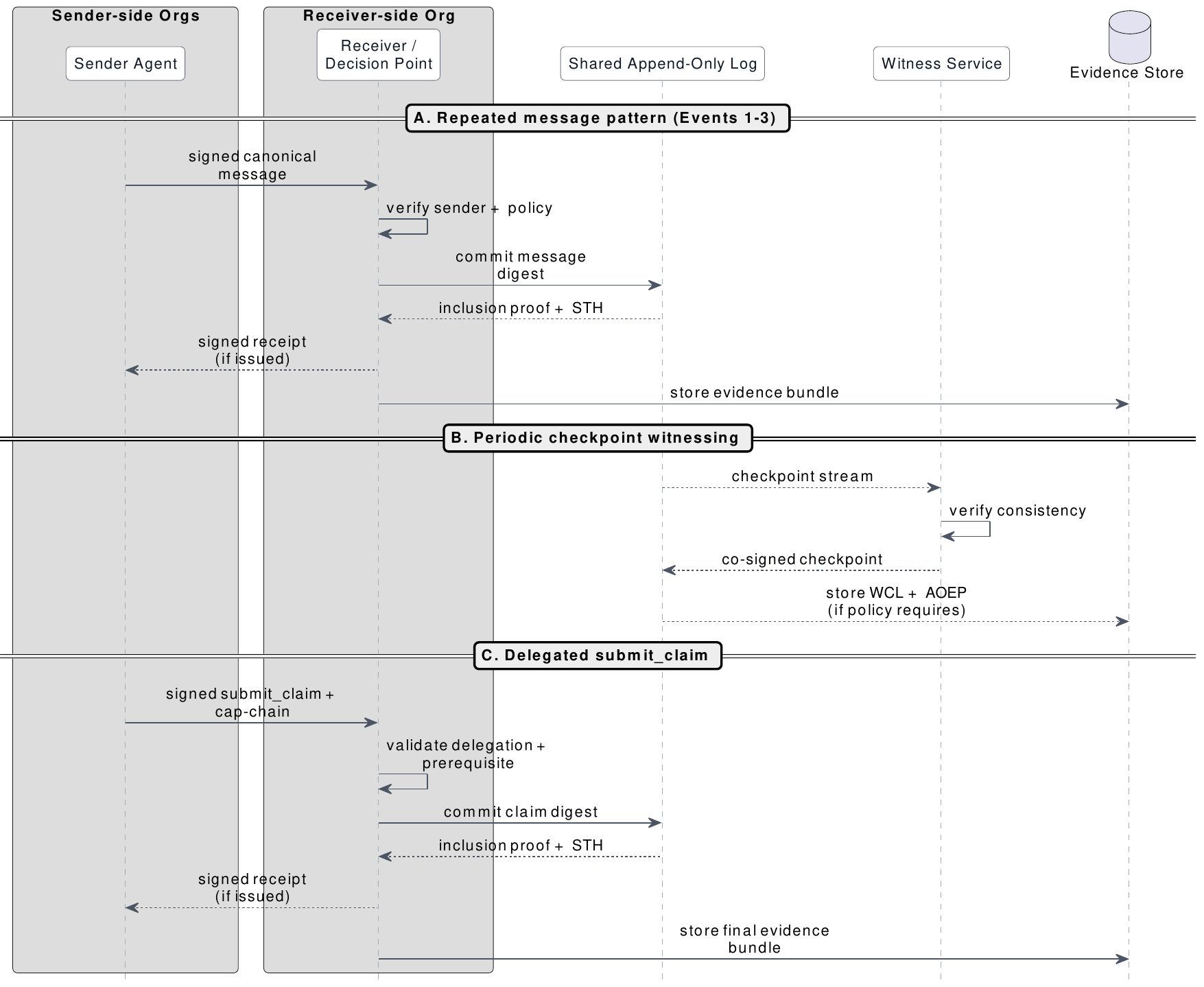}
  \caption{\textcolor{blue}{Conceptual evidence flow} for the
running workflow example. Event~1 is \texttt{refer\_patient},
Event~2 is \texttt{order\_mri}, Event~3 is \texttt{send\_result},
and Event~4 is \texttt{submit\_claim}. Events~1--3 share the
signed-message, validation, logging, and evidence-storage
pattern. Event~4 adds delegation and workflow-prerequisite
evidence. \textcolor{blue}{In the prototype, output containing
an admission receipt is returned after the bundle is written.}}

  \label{fig:workflow-sequence}
\end{figure*}

Issuance follows four stages. First, the system validates sender-authenticated inputs and any applicable
\textcolor{blue}{capability, receiver-context, or signed-reference}
conditions. Second, it commits the verified message digest to the
append-only log and obtains authenticated logging evidence.
Third, it \textcolor{blue}{completes the remaining workflow-state
checks and} attaches any additional evidence class required by the
active policy profile. Finally, it assembles and preserves the
evidence bundle for offline verification.
\textcolor{blue}{Any enabled receiver-admission receipt is constructed
and included before preservation, and the output is returned after
the bundle is successfully written. Some workflow-state refusals
occur after logging and leave the appended entry intact
(Section~\ref{subsec:receipt_eval}).}

The running example is a cross-organization healthcare workflow executed by agents on behalf of participating actors. A workflow event is a signed agent-to-agent message instance that represents a workflow action and is preserved with the evidence required for verification under the selected policy profile. In this example, Event~1 is \texttt{refer\_patient}, where a general practitioner (GP) agent refers a patient to a consultant. Event~2 is \texttt{order\_mri}, where a consultant-side agent requests an MRI from a lab or imaging provider. Event~3 is \texttt{send\_result}, where the lab or imaging agent returns the result. Event~4 is \texttt{submit\_claim}, where a claim-related follow-up action is submitted to the insurer. The example is illustrative. The verifier operates on preserved evidence bundles and does not depend on healthcare-specific semantics. 

Figure~\ref{fig:workflow-sequence} shows the
\textcolor{blue}{conceptual evidence flow} at the message-event
level. Events~1--3 share the repeated signed-message, validation, logging, and evidence-storage pattern, while Event~4 adds delegation and workflow-prerequisite evidence.

\textcolor{blue}{The witness-side consistency check shown in
Figures~\ref{fig:architecture} and~\ref{fig:workflow-sequence}
represents a deployment-design step that is not implemented in the prototype's witness-signing path.}

The canonical message representation $m$ is the deterministic
representation of the workflow message fields used for hashing and
signing. In the prototype, \textcolor{blue}{the canonical message}
contains an event type, case identifier,
\textcolor{blue}{receiver identifier, and event-specific workflow
fields. Examination identifiers, payload commitments, and
prior-evidence references are included where applicable. The sender
key identifier and signature are preserved separately in the bundle.}
This representation is used to verify sender authentication and to
derive the message commitment $h(m)$. The prototype evaluates
preserved evidence generated from these structured message events,
not a transport-level agent communication protocol.

Algorithm~\ref{alg:issuance} summarizes the issuance-side process that turns a validated message event into a preserved evidence bundle.

\begin{algorithm}[t]
\caption{Message processing and evidence bundle issuance}
\label{alg:issuance}
\footnotesize
\begin{algorithmic}[1]
\Require Canonical message representation $m$, sender evidence $\sigma_S$, policy profile $P$
\Ensure Preserved evidence bundle $\mathcal{B}$
\State Validate sender-authenticated inputs over the canonical message representation
\State Derive the message commitment $h(m)$
\State Commit the verified message digest to the append-only logging service
\State Obtain authenticated logging evidence, including inclusion proof and signed checkpoint
\State Apply any issuance-time workflow checks required by policy
\If{$P$ requires witness-backed checkpoint evidence}
    \State Obtain witness-backed checkpoint evidence
\EndIf
\If{$P$ requires checkpoint-context evidence}
    \State Preserve checkpoint-context evidence and associated WCL references
\EndIf
\If{$P$ requires append-only continuity validation and a prior checkpoint is available}
    \State Attach prior-checkpoint context and append-only extension evidence
\EndIf
\If{$P$ requires delegation-aware authorization evidence}
    \State Attach capability and policy-required workflow evidence
\EndIf
\If{an explicit signed receipt is issued and preserved}
    \State Attach explicit signed receipt evidence
\EndIf
\State Assemble the preserved evidence bundle $\mathcal{B}$
\State Persist $\mathcal{B}$ for later offline verification
\State \Return $\mathcal{B}$
\end{algorithmic}
\end{algorithm}

Algorithm~\ref{alg:issuance} follows the same issuance sequence: sender validation and message commitment in lines~1--2, logging evidence in lines~3--4, policy-required evidence attachment in lines~6--20, and bundle assembly and preservation in lines~21--22. Lines~6--20 also show why stronger later claims are not inferred retroactively: receipt, witness, continuity, checkpoint-context, and delegation evidence must be issued, available, and preserved when required by policy.

Evidence generation is therefore policy-aware at issuance time.
If a deployment requires stronger append-only continuity validation,
the runtime system must preserve the corresponding
\textcolor{blue}{authenticated checkpoints and AOEP. Witness
endorsement and WCL-context validation are separate checks.}
If it requires explicit receipt semantics, the receiver must issue
and preserve a receiver-signed receipt object rather than relying on
transport or storage side effects. If delegation-aware reasoning is
required, the delegation chain must be attached. The verifier later
checks these evidence classes under the selected policy.

\subsection{Evidence Bundle Structure and Verifier Inputs}
\label{subsec:bundle}
The preserved evidence bundle is the main accountability artifact evaluated by
the offline verifier. It contains core evidence for sender authentication and
authenticated log commitment. It may also contain policy-required evidence, such
as explicit receipt evidence, witness-backed checkpoint evidence, AOEP
continuity evidence, checkpoint-context material, or delegation/workflow
evidence. If the selected policy requires an artifact and that artifact is
absent or invalid, the corresponding verifier check fails.

The verifier evaluates the bundle using only local offline inputs: the selected
policy profile, public verification keys, trust anchors, key-resolution
references, and any checkpoint-context references permitted by the policy. It
does not require online access to the sender, receiver, logging service, or
witness. The system therefore separates evidence issuance during live message
processing from later evidence verification under the verifier's local control.

\section{Offline Verification and Policy-Controlled Evidence}
\label{sec:offline}

\subsection{Offline Verification Procedure}
\label{subsec:offline_verification}
\begin{figure*}
\centering
    \includegraphics[width=0.85\textwidth]{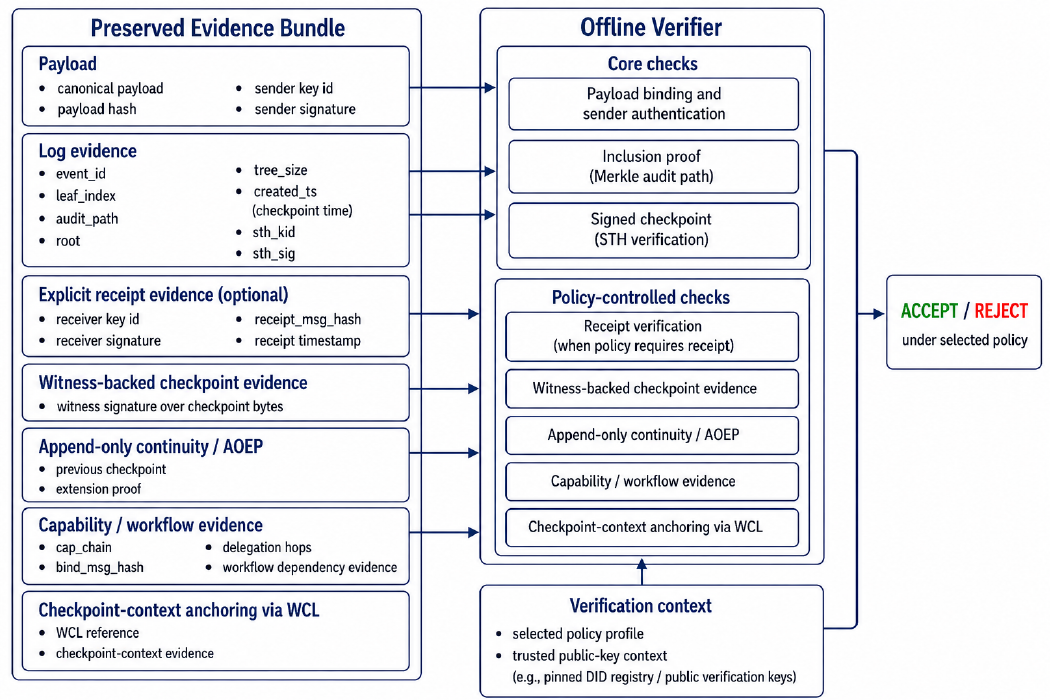}
    \caption{Offline verification inputs and policy-controlled checks over a preserved evidence bundle.}
    \label{fig:evidence_bundle_verifier}
\end{figure*}

\textcolor{blue}{The verifier checks a preserved evidence bundle using local verification keys, trust anchors, and a policy profile supplied under the assumptions in Section~\ref{subsec:assumptions}. In the prototype, the profile is loaded from the \texttt{WIT\_REQUIRED}, \texttt{CONS\_REQUIRED}, \texttt{CAP\_REQUIRED}, and \texttt{EPOCH\_REQUIRED} environment settings for witness evidence, append-only continuity, capabilities, and checkpoint context, respectively. A value of 1 selects the corresponding requirement; unset settings default to 0.}

\textcolor{blue}{The prototype can preserve issuance-time policy flags in a signed witnessed checkpoint log (WCL) record associated with checkpoint material. When checkpoint-context validation is required, the verifier checks record integrity and checkpoint membership without comparing the recorded flags with the active verification profile. The WCL therefore provides checkpoint-associated configuration evidence. Selection of the applicable policy and historical version follows the trust assumptions in Section~\ref{subsec:assumptions}. Acceptance under a different, weaker profile does not establish compliance with the policy originally agreed for the event.}

\textcolor{blue}{The verifier returns \textsc{Accept} only when all mandatory and policy-required checks succeed; otherwise, it returns \textsc{Reject}. Table~\ref{tab:policy_profiles} summarizes the five implemented profiles. Every profile retains checks of mandatory bundle fields, payload-hash binding, sender authentication, Merkle inclusion, and signed checkpoints.}

\begin{table}[!t]
\centering
\color{blue}
\caption{Policy profiles and verification requirements.}
\label{tab:policy_profiles}
\footnotesize
\renewcommand{\arraystretch}{1.12}
\setlength{\tabcolsep}{3pt}
\begin{tabularx}{\columnwidth}
{@{}l *{4}{>{\centering\arraybackslash}X}@{}}
\toprule
\textbf{Profile} &
\textbf{Witness} &
\textbf{Continuity} &
\textbf{Capability} &
\textbf{WCL} \\
\midrule
\texttt{ALL\_ON}   & 1 & 1 & 1 & 1 \\
\texttt{NO\_WIT}   & 0 & 1 & 1 & 1 \\
\texttt{NO\_AOEP}  & 1 & 0 & 1 & 1 \\
\texttt{NO\_CAP}   & 1 & 1 & 0 & 1 \\
\texttt{NO\_EPOCH} & 1 & 1 & 1 & 0 \\
\bottomrule
\end{tabularx}

\par\smallskip
\begin{minipage}{\columnwidth}
\scriptsize
\emph{Notes:}
1 enables the named requirement; 0 removes that requirement.
Continuity requires an extension proof when a previous checkpoint is supplied.
A supplied bundle witness signature is verified even under
\texttt{NO\_WIT}; WCL validation includes its own signature checks.
Claim-prerequisite verification remains active under \texttt{NO\_CAP}.
The separate version~1 correction also requires signed prior-evidence
binding under every profile (Section~\ref{subsec:tamper_eval}).

\textcolor{blue}{E10 varies an admission-receipt requirement independently of these five profiles. Missing admission receipts are permitted only when that requirement is disabled; supplied admission receipts are always verified, including those of recursively checked prerequisites (Section~\ref{subsec:receipt_eval}).}

\end{minipage}
\end{table}

\textcolor{blue}{The common checks support sender-authenticated message binding and authenticated log commitment. Successful witness, continuity, capability, and WCL checks additionally support checkpoint endorsement, authenticated extension between supplied checkpoints, delegated authority under the verified capability chain, and journal integrity with target-checkpoint membership, respectively. Receiver admission receipts are implemented and functionally evaluated separately in Section~\ref{subsec:receipt_eval}.}

\textcolor{blue}{Figure~\ref{fig:evidence_bundle_verifier} summarizes the bundle components and local inputs used in these checks. For a worked example, consider synthetic context~1 from the separate revision study in Section~\ref{subsec:tamper_eval}. Under \texttt{ALL\_ON}, the preserved \texttt{submit\_claim} bundle and its intended \texttt{send\_result} evidence were accepted by the version~1 verifier. The accepted path checks the signed message, payload binding, checkpoint authentication and inclusion, required witness and continuity evidence, WCL context, workflow prerequisites, and the capability chain. Keeping the signed claim unchanged while resolving an earlier, independently valid result from the same case and examination instead produced \textsc{Reject} with \texttt{prior\_reference\_mismatch}. The claim's signed reference committed to the intended result, so the alternative result did not satisfy that reference despite being valid in its own context.}

\begin{algorithm}[!t]
\caption{Offline verification of an evidence bundle}
\label{alg:offline}
\footnotesize
\begin{algorithmic}[1]
\Require Evidence bundle $\mathcal{B}$, offline verification inputs, policy profile $P$
\Ensure \textsc{Accept} or \textsc{Reject} with the first failing check
\State Validate bundle structure and mandatory fields
\State Verify payload-hash binding for the preserved canonical message representation
\State Verify sender-authentication evidence $\sigma_S$ against the canonical message and reject if absent or invalid
\State Verify authenticated log commitment via inclusion proof and signed checkpoint and reject if verification fails
\If{$P$ requires witness-backed checkpoint validation}
    \State Verify witness evidence against the authenticated checkpoint
\EndIf
\If{$P$ requires checkpoint-context validation}
    \State Verify checkpoint-context evidence against the permitted WCL-backed context
\EndIf
\If{$P$ requires append-only continuity validation}
    \State Verify prior-checkpoint context and AOEP evidence
\EndIf
\If{$P$ requires delegation or policy-required workflow checks}
    \State Verify capability and policy-required workflow evidence
\EndIf
\If{$P$ requires explicit receipt}
    \State Verify explicit signed receipt evidence
\EndIf
\If{all mandatory and policy-required checks succeed}
    \State \Return \textsc{Accept}
\Else
    \State \Return \textsc{Reject} with the first failing check
\EndIf
\end{algorithmic}
\end{algorithm}

\textcolor{blue}{Algorithm~\ref{alg:offline} summarizes the evidence-model checks. Table~\ref{tab:policy_profiles} specifies the implemented profile conditions and applicable checks. Definition~\ref{def:policy-relative-acceptance} defines policy-relative acceptance within the evidence model.}

\subsection{Append-Only Continuity Verification}
\label{subsec:aoep}

Algorithm~\ref{alg:aoep} summarizes the verifier-side AOEP check performed when the selected policy requires append-only continuity between two authenticated checkpoints.

\begin{algorithm}
\caption{AOEP verification between two checkpoints}
\label{alg:aoep}
\footnotesize
\begin{algorithmic}[1]
\Require Prior authenticated checkpoint $\mathrm{STH}_{old}$, later authenticated checkpoint $\mathrm{STH}_{new}$, extension proof $\pi_{\mathrm{ext}}$
\Ensure Valid append-only extension or failure
\State Verify the authenticity of $\mathrm{STH}_{old}$ and $\mathrm{STH}_{new}$
\State Check that checkpoint ordering and applicability conditions hold
\State Verify that $\pi_{\mathrm{ext}}$ validates append-only extension from $\mathrm{STH}_{old}$ to $\mathrm{STH}_{new}$
\If{all checks succeed}
    \State \Return valid append-only extension
\Else
    \State \Return failure
\EndIf
\end{algorithmic}
\end{algorithm}

\textcolor{blue}{Algorithm~\ref{alg:aoep} verifies the checkpoint signatures, ordering conditions, and append-only extension proof. Its conclusion concerns extension between the supplied authenticated checkpoints. Consistency across all observers and resistance to log--witness collusion remain outside this check; the general claim limits are stated in Section~\ref{subsec:corollary-interpretation}.}

\subsection{Delegation-Aware Authorization Verification}
\label{subsec:cap}

\textcolor{blue}{Algorithm~\ref{alg:cap} checks the capability chain when the selected policy requires delegation-aware authorization. This evidence check is separate from sender authentication and log inclusion.}

\begin{algorithm}
\caption{Capability-chain verification}
\label{alg:cap}
\footnotesize
\begin{algorithmic}[1]
\Require Delegation evidence $\mathcal{C}_{\mathrm{cap}}$, message commitment $h(m)$, policy profile $P$
\Ensure Valid delegation evidence or failure
\State Check actor-target consistency across delegation links
\State Check audience progression and parent-link consistency
\State Check scope and constraint consistency
\State Check temporal validity of delegation artifacts
\State Check message binding to the delegation context when required
\State Verify signatures or authenticators on delegation artifacts
\If{all required checks succeed}
    \State \Return valid delegation evidence
\Else
    \State \Return failure
\EndIf
\end{algorithmic}
\end{algorithm}

\begin{figure*}
    \centering
    \includegraphics[width=0.85\textwidth]{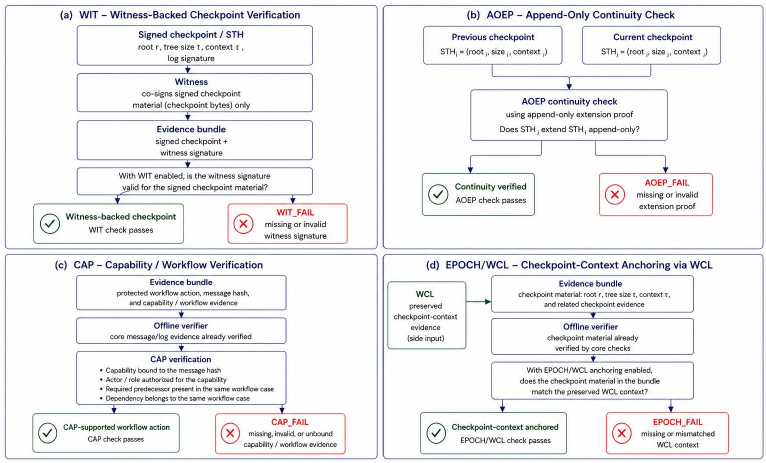}
    \caption{Policy-controlled evidence checks used by the fully enabled verifier profile and the single-control policy comparison.}
    \label{fig:policy_controlled_checks}
\end{figure*}

\textcolor{blue}{Figure~\ref{fig:policy_controlled_checks} shows the policy-controlled checks used by the fully enabled verifier profile and the single-control policy comparison. In panel (d), EPOCH/WCL denotes checkpoint-context validation. The figure's anchoring labels refer to this same check: verifying the integrity of the preserved witnessed checkpoint log (WCL) and membership of the target checkpoint. The \texttt{EPOCH\_REQUIRED} setting enables this requirement, and the \texttt{NO\_EPOCH} profile omits it.}


\section{Acceptance Semantics and Bounded Claims}
\label{sec:semantics}

\textcolor{blue}{This section defines policy-relative acceptance and claim support for preserved evidence bundles. The propositions and corollary state logical consequences of these definitions and explain how to interpret the verifier's decision. Interpreting successful checks as authenticated evidence relies on the cryptographic and trust assumptions in Section~\ref{subsec:assumptions}. These results concern the evidence model and do not constitute cryptographic reduction proofs for the full runtime system.}

We use the notation from Section~\ref{subsec:notation}. Let $\mathcal{B}$ be the preserved evidence bundle for one message event, and let $P$ be the selected policy profile. Let $CP(P)$ be the set of verification predicates required under $P$. When the policy is clear, we write $CP$.

We use $\varphi$ for verifier predicates to distinguish them from proof artifacts such as $\pi_{\mathrm{incl}}$ and $\pi_{\mathrm{ext}}$. The verifier predicates considered in this paper are
\begin{equation}
\Phi =
\lbrace
\varphi_{\mathrm{fmt}},
\varphi_{\mathrm{hash}},
\varphi_{\mathrm{sig}},
\varphi_{\mathrm{incl}},
\varphi_{\mathrm{sth}},
\varphi_{\mathrm{wit}},
\varphi_{\mathrm{ctx}},
\varphi_{\mathrm{aoep}},
\varphi_{\mathrm{cap}},
\varphi_{\mathrm{rcpt}}
\rbrace .
\label{eq:predicate-universe}
\end{equation}

These predicates correspond to bundle format validity, payload-hash binding, sender-authenticated message binding, Merkle inclusion verification, signed-checkpoint verification, witness-backed checkpoint validation, checkpoint-context validation, append-only extension validation, capability and workflow validation, and explicit signed-receipt validation.

A policy profile selects the predicates required for acceptance:
\begin{equation}
CP(P) \subseteq \Phi .
\label{eq:policy-predicate-set}
\end{equation}

All verifier profiles require the following core predicates:
\begin{equation}
CP_{\mathrm{core}}
=
\lbrace
\varphi_{\mathrm{fmt}},
\varphi_{\mathrm{hash}},
\varphi_{\mathrm{sig}},
\varphi_{\mathrm{incl}},
\varphi_{\mathrm{sth}}
\rbrace
\subseteq CP(P).
\label{eq:core-predicates}
\end{equation}

Thus, authenticated log commitment is checked through both inclusion verification and signed-checkpoint verification. Stronger profiles may also require witness evidence, checkpoint-context evidence, AOEP continuity evidence, capability/workflow evidence, and explicit receipt evidence.

\subsection{Acceptance Conditions}
\label{subsec:acceptance-conditions}

\begin{definition}[Policy-relative acceptance]
\label{def:policy-relative-acceptance}
For an evidence bundle $\mathcal{B}$ and policy profile $P$, acceptance is defined as
\begin{equation}
\mathrm{Accept}_P(\mathcal{B})
\iff
\bigwedge_{\varphi \in CP(P)}
\left(\varphi(\mathcal{B}) = 1\right).
\label{eq:acceptance}
\end{equation}
Thus, $\mathcal{B}$ is accepted only if every mandatory and policy-required predicate succeeds.
\end{definition}

The preserved message commitment is represented as
\begin{equation}
h(m) = H(m),
\label{eq:message-commitment}
\end{equation}
where $m$ is the preserved canonical message representation and $H(\cdot)$ is the hash or commitment function.

An authenticated checkpoint is represented abstractly as a signed tree head:
\begin{equation}
\sigma_L =
\mathrm{Sign}_L(r \mathbin{\|} t \mathbin{\|} \tau),
\qquad
\mathrm{STH} = (r,t,\tau,\sigma_L),
\label{eq:signed-tree-head}
\end{equation}

where $r$ is the authenticated root, $t$ is the
\textcolor{blue}{Merkle-tree size}, $\tau$ is the checkpoint context,
and $\sigma_L$ is the log signature. This is a semantic
representation. It does not prescribe a byte layout.

\subsection{Claim Support}
\label{subsec:claim-support}

Acceptance and claim support are different. A bundle may be accepted under a policy, but it supports only the claims justified by the predicates required by that policy.

Let $\mathsf{Claim}$ be the set of evidence-level claims considered by the verifier. For each claim $c \in \mathsf{Claim}$, let $\mathsf{Req}(c) \subseteq \Phi$ be the predicates required to support that claim.

\begin{definition}[Policy-relative claim support]
\label{def:policy-relative-claim-support}
A bundle $\mathcal{B}$ supports claim $c$ under policy $P$ only if all predicates required for $c$ are required by $P$ and succeed on $\mathcal{B}$:
\begin{equation}
\mathrm{Supports}_P(\mathcal{B},c)
\iff
\mathsf{Req}(c) \subseteq CP(P)
\ \wedge\
\bigwedge_{\varphi \in \mathsf{Req}(c)}
\left(\varphi(\mathcal{B}) = 1\right).
\label{eq:claim-support}
\end{equation}
\end{definition}

For an authenticated logged-message claim, the required predicates are
\begin{equation}
\mathsf{Req}(c_{\mathrm{logged\_msg}})
=
\lbrace
\varphi_{\mathrm{fmt}},
\varphi_{\mathrm{hash}},
\varphi_{\mathrm{sig}},
\varphi_{\mathrm{incl}},
\varphi_{\mathrm{sth}}
\rbrace .
\label{eq:logged-message-requirements}
\end{equation}

\textcolor{blue}{This claim binds the preserved message to sender-authentication evidence and an authenticated log checkpoint. Its scope is interpreted using Section~\ref{subsec:corollary-interpretation}.}

A receipt claim requires explicit receiver-signed receipt evidence:
\begin{equation}
\mathsf{Req}(c_{\mathrm{receipt}})
=
\lbrace
\varphi_{\mathrm{fmt}},
\varphi_{\mathrm{hash}},
\varphi_{\mathrm{sig}},
\varphi_{\mathrm{incl}},
\varphi_{\mathrm{sth}},
\varphi_{\mathrm{rcpt}}
\rbrace .
\label{eq:receipt-requirements}
\end{equation}

\begin{definition}[Explicit receipt predicate]
\label{def:explicit-receipt-predicate}
\textcolor{blue}{In the evidence model, a receipt is the receiver-side decision point's signed assertion that it has admitted the identified message for workflow processing under the applicable policy. Issuance is triggered by an explicit admission decision after sender authentication and the applicable intake checks succeed. The signed assertion identifies this admission state, the workflow event, the message commitment, the receiver, the policy identifier and version, and the issuance time. Successful receipt verification authenticates this specific receiver assertion.}

\textcolor{blue}{This admission-receipt contract is implemented and functionally evaluated in a separate revision candidate, using trusted external receiver-key and admission-policy inputs. The receipt study is separate from Experiments~E1--E9 and their reported measurements.}

When evaluated, the predicate $\varphi_{\mathrm{rcpt}}(\mathcal{B})$ evaluates to $1$ only if $\mathcal{B}$ contains an explicit receiver-signed receipt object and that object verifies under the applicable receiver key, message commitment, and policy context. Otherwise,
\begin{equation}
\varphi_{\mathrm{rcpt}}(\mathcal{B}) = 0 .
\label{eq:receipt-predicate-zero}
\end{equation}
\end{definition}

\subsection{Bounded Interpretation of \textsc{Accept}}
\label{subsec:bounded-accept}

\begin{proposition}[Bounded interpretation of \textsc{Accept}]
\label{prop:bounded-accept}
If $\mathrm{Accept}_P(\mathcal{B})$ holds, then every predicate in $CP(P)$ has succeeded for $\mathcal{B}$. The accepted bundle supports only the evidence-level claims whose required predicates are included in $CP(P)$ and have succeeded on $\mathcal{B}$. No stronger claim follows from \textsc{Accept} alone.
\end{proposition}

\begin{proof}
By Definition~\ref{def:policy-relative-acceptance}, $\mathrm{Accept}_P(\mathcal{B})$ means that all predicates in $CP(P)$ evaluate to $1$. By Definition~\ref{def:policy-relative-claim-support}, a claim $c$ is supported only when $\mathsf{Req}(c) \subseteq CP(P)$ and all predicates in $\mathsf{Req}(c)$ evaluate to $1$. Therefore, acceptance supports only the claims whose required predicates are satisfied under the selected policy. Claims requiring predicates outside $CP(P)$ do not follow from \textsc{Accept}.
\end{proof}

\begin{proposition}[Policy monotonicity under fixed predicate semantics]
\label{prop:policy-monotonicity}
Let $P_1$ and $P_2$ be two policy profiles evaluated with the same bundle representation, trust anchors, verification keys, verifier inputs, and predicate definitions. If
\begin{equation}
CP(P_1) \subseteq CP(P_2),
\label{eq:policy-subset}
\end{equation}
then, for any evidence bundle $\mathcal{B}$,
\begin{equation}
\mathrm{Accept}_{P_2}(\mathcal{B})
\Rightarrow
\mathrm{Accept}_{P_1}(\mathcal{B}).
\label{eq:policy-monotonicity}
\end{equation}
The converse does not necessarily hold.
\end{proposition}

\begin{proof}
If $\mathrm{Accept}_{P_2}(\mathcal{B})$ holds, then every predicate in $CP(P_2)$ succeeds. Since $CP(P_1) \subseteq CP(P_2)$, every predicate required by $P_1$ also succeeds. Therefore, $\mathrm{Accept}_{P_1}(\mathcal{B})$ holds. The converse may fail because $P_2$ may require additional predicates.
\end{proof}

\textcolor{blue}{Proposition~\ref{prop:policy-monotonicity} applies to adding required predicates while holding the verification conditions fixed. These conditions include the bundle, verification keys, trust anchors, freshness and key-resolution rules, predicate definitions, and offline context. The proposition does not compare executions in which those conditions also change.}

\subsection{Corollary and Interpretation}
\label{subsec:corollary-interpretation}

\begin{corollary}
\label{cor:no-receipt-no-claim}
Let $\mathcal{B}$ be evaluated under policy profile $P$. If $\varphi_{\mathrm{rcpt}} \notin CP(P)$ or $\varphi_{\mathrm{rcpt}}(\mathcal{B}) = 0$, then $\mathcal{B}$ does not support a receipt claim under $P$:
\begin{equation}
\neg \mathrm{Supports}_P(\mathcal{B},c_{\mathrm{receipt}}).
\label{eq:no-receipt-no-claim}
\end{equation}
\end{corollary}

\begin{proof}
By Equation~\ref{eq:receipt-requirements}, a receipt claim requires $\varphi_{\mathrm{rcpt}}$. By Definition~\ref{def:policy-relative-claim-support}, the claim is supported only if all required predicates are required by $P$ and succeed on $\mathcal{B}$. Therefore, if $\varphi_{\mathrm{rcpt}}$ is not required by $P$, or if it evaluates to $0$, the receipt claim is unsupported.
\end{proof}

\textcolor{blue}{A receipt claim requires $\varphi_{\mathrm{rcpt}}$; transport outcomes and other evidence classes cannot replace this predicate. In the evidence model, successful receipt verification authenticates the preserved receiver assertion, without establishing delivery, human reading, semantic understanding, or successful downstream execution. More generally, \textsc{Accept} supports only the policy-relative claims defined above. It does not establish endpoint execution correctness, semantic correctness, workflow completion, universal consistency, real-time equivocation prevention, or protection against runtime LLM threats.}


\section{Evidence-Level Security Assessment}
\label{sec:security_analysis}
Section~\ref{sec:semantics} defined the formal acceptance and claim-support semantics. This section assesses evidence-level security threats by mapping them to the preserved evidence they may invalidate, omit, modify, or rebind. It does not provide a separate cryptographic reduction proof or a full runtime-security analysis.

\subsection{Verifier-Supported Security Properties} \label{subsec:properties}

Table~\ref{tab:verifier_properties} summarizes the verifier-supported evidence properties and the evidence required for each property. These properties are derived from preserved evidence checks under the selected policy profile.

\begin{table*}[!t]
\centering
\caption{Verifier-supported security properties and required evidence.}
\label{tab:verifier_properties}
\footnotesize
\renewcommand{\arraystretch}{1.05}
\begin{tabularx}{\textwidth}{C{0.06\textwidth} L{0.30\textwidth} X}
\toprule
\textbf{ID} & \textbf{Property} & \textbf{Required evidence and scope} \\
\midrule
P1 & Authenticated log commitment &
Inclusion proof and signed checkpoint; verified checkpoint only. \\

P2 & Sender-authenticated message binding &
Sender-authenticated message fields and payload-to-log binding evidence. \\

P3 & Append-only continuity &
Prior-checkpoint context and append-only extension evidence; policy-enabled. \\

P4 & Witness-backed checkpoint validity &
Witness evidence; does not imply receipt, delivery, or real-time equivocation prevention. \\

P5 & Capability and policy-required workflow checks &
Delegation-aware authorization evidence and required workflow evidence. \\

P6 & Checkpoint-context validation &
Checkpoint-context evidence and permitted WCL context; policy-enabled. \\

\bottomrule
\end{tabularx}
\end{table*}

\textcolor{blue}{P3 concerns append-only extension between authenticated checkpoints, P4 witness evidence, and P6 checkpoint context. Receipt support requires the predicate defined in Definition~\ref{def:explicit-receipt-predicate}. The properties are interpreted under the acceptance and claim-support semantics in Section~\ref{sec:semantics}.}

\subsection{Evidence-Level Security Threats and Mapping}
\label{subsec:evidence_failure_mapping}

\begin{table*}[!t]
\centering
\caption{Evidence-level security threats mapped to verifier-supported properties.}
\label{tab:evidence_failure_mapping}
\footnotesize
\renewcommand{\arraystretch}{1.08}
\begin{tabularx}{\textwidth}{C{0.04\textwidth} L{0.17\textwidth} L{0.19\textwidth} X C{0.08\textwidth}}
\toprule
\textbf{ID} & \textbf{Security threat} & \textbf{Affected evidence} & \textbf{Evidence issue} & \textbf{Props.} \\
\midrule

T1 & Identity spoofing &
Signed-message submission &
False or impersonated sender identity. & P2 \\

T2 & Log tampering &
Log evidence &
Inclusion proof, signed checkpoint, or log material modified. & P1 \\

T3 & Rebinding &
Payload-to-log binding &
Commitment, leaf index, or Merkle path modified or rebound. & P1, P2 \\

T4 & Repudiation &
Sender evidence &
Sender denies authorship despite preserved sender authentication. & P2 \\

T5 & Tampering or omission &
Witness/AOEP/context &
Required evidence missing or modified. & P3, P4, P6 \\

T6 & Unauthorized delegation &
Workflow/delegation &
Capability or workflow evidence forged, rebound, or missing. & P5 \\

\bottomrule
\end{tabularx}
\end{table*}

Table~\ref{tab:evidence_failure_mapping} links each evidence-level security threat to the affected evidence class and the corresponding verifier-supported property. The mapping is evidence-level rather than deployment-wide: it
shows how identity spoofing, log tampering, rebinding, repudiation, tampering or omission, and unauthorized delegation may affect the evidence checked by the verifier.

The mapping provides the basis for the targeted negative-evidence cases evaluated later.

\section{Evaluation}
\label{sec:evaluation}
\subsection{Evaluation Scope and Methodology}
\label{subsec:evaluation_scope}

\textcolor{blue}{Table~\ref{tab:evaluation_design}
summarizes the design and purpose of the experiments.}
The performance part of the evaluation measures offline verifier-side latency, and the diagnostic part measures rejection behavior over preserved evidence bundles. It does not measure memory utilization, network delay, delivery, runtime LLM safety, end-to-end orchestration, or workflow throughput. 

\begin{table*}
\centering
\begingroup
\color{blue}
\caption{Design and purpose of the experiments.}
\label{tab:evaluation_design}
\footnotesize
\renewcommand{\arraystretch}{1.08}
\setlength{\tabcolsep}{5pt}

\begin{tabularx}{\textwidth}
{@{}C{0.06\textwidth}
L{0.22\textwidth}
L{0.31\textwidth}
X@{}}
\toprule
\textbf{Exp.} &
\textbf{Evaluation question} &
\textbf{Method} &
\textbf{Interpretation} \\
\midrule

E1 &
How does verifier latency vary across policy profiles? &
Comparison of five profiles over valid
\texttt{submit\_claim} bundles. &
Conditional latency differences relative to
\texttt{ALL\_ON}. \\

E2 &
How do event types differ in latency and bundle size? &
Paired event-type latency comparison and
serialized bundle-size analysis. &
Event-specific verification and storage requirements. \\

E3 &
Does invalid or policy-insufficient evidence fail
at the expected check? &
Targeted negative-evidence cases mapped to
checks, threats, and properties. &
Diagnostic rejection coverage, not exhaustive
adversarial search. \\

E4 &
Does execution pattern affect claim latency
under matched inputs? &
Claims-only and four-event execution using
the same bundles, profile, and WCL. &
Paired latency comparison within the recorded run. \\

E5 &
Which stages contribute to checkpoint-context
processing time? &
Instrument WCL processing at three journal lengths. &
Internal stage times for the instrumented verifier. \\

E6 &
How does the verifier handle individually valid
evidence presented in an incorrect context? &
Correct-context controls and valid-donor
substitutions under five profiles. &
Profile-specific substitution behavior in
the tested verifier. \\

E7 &
Is a claim bound to its committed prior evidence? &
Compare baseline and corrected verifiers;
add corrected-verifier-only schema and regression tests. &
Reference-binding behavior and compatibility
by verifier version. \\

E8 &
Does corrupted input produce a controlled rejection? &
Replay a fixed corruption corpus before and after
the input-handling repair, with valid controls. &
Controlled rejection within the tested
corruption classes. \\

E9a--E9d &
How does latency vary with evidence and workflow size? &
Separate studies varying history/WCL length,
signed-message size, delegation-chain length,
and workflow-event count. &
Sensitivity within each study's tested range. \\

\textcolor{blue}{E10} &
\textcolor{blue}{Can admission receipts be issued and verified
with the required context bindings?} &
\textcolor{blue}{Correct-context controls and receipt mutations
under five profiles with receipts optional or required;
actual issuance tests in a candidate extending the
implementation corrected in E8.} &
\textcolor{blue}{Authenticity and context binding of the
receiver's admission assertion.} \\

\bottomrule
\end{tabularx}

\par\smallskip
\begin{minipage}{0.98\textwidth}
\scriptsize
\emph{Note:}
Each experiment specifies its verifier version,
inputs, and procedure in its accompanying description.
E9a--E9d are separate studies; retained history and
WCL length vary together in E9a.
\end{minipage}
\endgroup
\end{table*}

\textcolor{blue}{Experiment E2 evaluates four event types:
\texttt{refer\_patient}, \texttt{order\_mri},
\texttt{send\_result}, and \texttt{submit\_claim}.}

These are the same four events introduced earlier in the running
cross-organization healthcare scenario shown in
Figure~\ref{fig:workflow-sequence}. The labels are used as a running
example. Verification uses preserved evidence bundles and does not
depend on healthcare-specific semantics.

\textcolor{blue}{E1 and E2 were rerun using a shared corpus of
300 synthetic four-event workflow cases generated by the
prototype's issuance procedure, producing 1200 preserved bundles.
The cases have distinct case and examination identifiers.
Six logical roles---GP, consultant, laboratory, insurer, log,
and witness---use role-specific signing keys reused across cases.
The first three event types carry one capability token each.
Each \texttt{submit\_claim} carries a two-hop capability chain
and a signed reference to its prerequisite result.
All measurement conditions use the same public-key registry
and the same frozen final WCL containing 1200 records.}

\textcolor{blue}{Workload seed 20261029 determines only synthetic
content and nonce bytes. Native event identifiers, timestamps,
capability identifiers, and signing keys are generated independently
of that seed. Exact replay uses the preserved public artifacts.}

\textcolor{blue}{In addition to the core verifier checks, the fully enabled profile requires witness validation, AOEP continuity validation, delegation/workflow checks, and checkpoint-context validation.}

\subsection{\textcolor{blue}{Implementation and Experimental Environment}}
\label{subsec:experimental_environment}

\textcolor{blue}{The revised E1 and E2 measurements used Windows build 10.0.26300, an Intel Core i7-13700H processor, 20 logical processors, 32 GB RAM, and Python 3.9.13 (64-bit). E3 retains its original diagnostic results from the reported Microsoft Surface Laptop Studio 2 environment running Windows 11 Home build 10.0.26200 and Python 3.9.13. Agent execution, manual inspection, and evidence generation
were outside the timed intervals. Experiment E4 reports its own runtime conditions in Section~\ref{subsec:unified_timing}.}

\textcolor{blue}{The prototype uses SHA-256 for message commitments and record hashes, and Ed25519 detached signatures through PyNaCl. Ed25519 public keys and signing seeds contain 32 bytes, and signatures contain 64 bytes. Issuers represent message and record digests in hexadecimal and encode public keys, signatures, Merkle roots, and proof nodes using unpadded Base64url. The evaluated path provides integrity and authentication through hashing and digital signatures; payload encryption is outside this verification procedure.}

\textcolor{blue}{Message cores and capability statements use deterministic UTF-8 JSON with sorted object keys, compact comma/colon separators, and \texttt{ensure\_ascii=False}. The verifier authenticates the preserved message bytes directly. The separately tested revision candidate requires the canonical message, sender key identifier, and sender signature, and checks both the message commitment and sender signature under every profile, including recursive prerequisite verification. It derives workflow decisions from authenticated message fields and checks agreement with the corresponding cached metadata. Signed prior-evidence binding is evaluated separately in Section~\ref{subsec:tamper_eval}.}

\textcolor{blue}{Let $H$ denote SHA-256 and $m$ the preserved message bytes. The Merkle leaf hash is $H(\texttt{0x00}\mathbin{\|}H(m))$, and an internal-node hash is $H(\texttt{0x01}\mathbin{\|}L\mathbin{\|}R)$, where \texttt{0x00} and \texttt{0x01} are single-byte prefixes. Inclusion proofs contain sibling hashes ordered from leaf to root. The default \texttt{MERKLE\_DUP\_LAST=0} promotes an unmatched final node; setting it to 1 selects duplication. AOEP generation recursively splits subtrees at the largest power of two below their size. Verification checks the supplied tree-size relationships and reconstructs the roots associated with the signed checkpoints.}

\textcolor{blue}{Checkpoint signing bytes concatenate the raw Merkle root, an unsigned eight-byte big-endian tree size, and the UTF-8 creation timestamp. Issuance uses UTC timestamps in \texttt{YYYY-MM-DDTHH:MM:SSZ} format. The witness signs the same checkpoint bytes with the prefix \texttt{WCL-STH-v1|}. Version 1 WCL records use JSON Lines. The hashed record body contains \texttt{v}, \texttt{checkpoint\_seq}, \texttt{created}, \texttt{prev\_rec\_hash\_hex}, \texttt{did\_registry\_hash\_hex}, \texttt{policy}, \texttt{sth}, \texttt{witness}, and \texttt{sth\_msg\_b64u}. This body uses deterministic UTF-8 JSON with sorted keys, compact comma/colon separators, and \texttt{ensure\_ascii=False}. Its digest is SHA-256 over \texttt{WCL-REC-v1|} followed by these serialized body bytes. The record signature covers \texttt{WCL-REC-SIG-v1|} followed by the raw digest. The resulting \texttt{rec\_hash\_hex} and \texttt{epoch\_sig\_b64u} fields are stored alongside the hashed body.}

\textcolor{blue}{Checkpoint-context verification scans the complete supplied WCL, checking its sequence, record chaining, hashes, and witness signatures, and locating the target checkpoint signing bytes. The scan continues after finding the target. With checkpoint-context verification enabled, an accepted \texttt{submit\_claim} verifies both its own bundle and its prerequisite bundle, producing two WCL scans. Policy loading and the role of preserved configuration evidence are described in Section~\ref{subsec:offline_verification}.}

\textcolor{blue}{Each capability hop contains a signed statement identifying the issuer key, subject, audience, scope, issuance and expiry times, token identifier, and message digest. Chained hops include the preceding statement's SHA-256 digest. Verification checks signatures, message binding, subject and audience relationships, parent hashes, and applicable action, case, and examination scope values. Scope attenuation requires each child field and value to occur unchanged in the parent. Capability-time checks use the signed checkpoint creation time: expiry before that reference is rejected, and issuance up to 60 seconds after it is permitted. \textcolor{blue}{E1/E2 used positive signed reference times.} A nonpositive parsed reference time selects the local clock.}

\textcolor{blue}{For message, log, witness and capability checks, public keys are resolved locally by key identifier from the registry preserved in the bundle, using its \texttt{x\_b64u} or \texttt{x} entry. Registry authenticity and historical suitability follow the trust assumptions in Section~\ref{subsec:assumptions}. The prototype uses separately configured log and witness signing keys. The revised E1 and E2 measurements use the corrected verifier evaluated in E10, with receiver-admission receipt issuance and verification requirements disabled and no admission receipts in the performance bundles. E3 and E4 retain their historical verifier versions. The separate functional studies are reported in Sections~\ref{subsec:tamper_eval} and~\ref{subsec:receipt_eval}.}

\textcolor{blue}{E10 preserves a separate \texttt{receiver\_admission\_receipt} alongside the logging \texttt{receipt}. Ed25519 signs the prefix \texttt{A2A-RECEIVER-ADMISSION-v1|} followed by the deterministically serialized receipt body. This body binds the admission state, event label, a reference to the authenticated message and its logged occurrence, receiver and key identifiers, admission-policy identifier and version, and issuance time. Approved receiver keys and the expected admission policy and version are supplied through a trusted external context configured by \texttt{RECEIVER\_ADMISSION\_CONTEXT\_PATH}. The independent settings \texttt{RECEIVER\_ADMISSION\_ISSUE} and \texttt{RECEIVER\_ADMISSION\_REQUIRED} control issuance and the verification requirement, respectively; both default to 0. Admission signing follows successful applicable intake and workflow checks, even when logging has already occurred. Issuance timestamps are checked for valid UTC syntax.}

\textcolor{blue}{The preserved dependency record lists PyNaCl 1.6.2, FastAPI 0.128.0, Uvicorn 0.39.0, Pydantic 2.12.5, python-dotenv 1.2.1, NumPy 2.0.2, pandas 2.3.3, SciPy 1.13.1, Matplotlib 3.9.4, and requests 2.32.5. Hashing and JSON processing use Python's \texttt{hashlib} and \texttt{json} modules.}

\subsection{Statistical Treatment}
\label{subsec:statistical_treatment}

\textcolor{blue}{Each timed invocation launches a fresh isolated Python process. A monotonic timer surrounds the complete subprocess call, including interpreter and launcher startup, imports, evidence access, verification, stdout/stderr capture, and process exit. Controller-side provenance and output hashing, captured-output decoding, and result recording occur outside this interval. Each timed call has a 60-second timeout.}

\textcolor{blue}{Measurements span three cycles. In each cycle, E1 verifies each case's claim under five profiles, while E2 verifies its four events consecutively under \texttt{ALL\_ON}. The five profile conditions and the four-event sequence follow a prespecified counterbalanced order; case order is randomized with seed 42. The standalone and sequence-context claims are timed independently. This gives 8100 timed calls. A separate 2400-call valid-input gate and 27 warm-up calls, nine before each cycle, are excluded from the latency analysis.}

\textcolor{blue}{Each condition contributes 300 case medians, each calculated from three measured invocations. Primary summaries report the median and P95 of these case medians, with P95 calculated using linear interpolation. Separate summaries retain the median and P95 across the 900 individual calls per condition. All timed observations are retained.}

\textcolor{blue}{The 95\% confidence intervals use a percentile bootstrap with 3000 resamples of 300 cases with replacement and seed 42. The same case-index draws preserve pairing across conditions. Bounds are the linearly interpolated 2.5th and 97.5th percentiles of the resampled medians. The matched claim comparison reports the median of the 300 case-level sequence-minus-standalone \texttt{ALL\_ON} differences, with an interval obtained using the same case-bootstrap procedure.}

E1 uses four two-sided paired Wilcoxon signed-rank tests against standalone \texttt{ALL\_ON} \textcolor{blue}{to compare profiles within the same cases}, with Holm correction \textcolor{blue}{for multiple comparisons}. E2 uses the Friedman test \textcolor{blue}{to compare the four event types within each case}, Kendall's \(W\) \textcolor{blue}{to summarize the rank-based effect size}, and six two-sided paired Wilcoxon comparisons \textcolor{blue}{to identify event-pair differences}, forming a separate Holm family. E3 retains its original diagnostic evaluation of rejection behavior.

\textcolor{blue}{OS caches are not flushed; the validation gate and warm-ups access the inputs before measurement. Garbage collection uses interpreter defaults, while CPU frequency and unrelated background activity are not experimentally controlled. External power availability is checked at phase boundaries, and a temporary system-awake request covers execution. Power plans, process priority and CPU affinity are unchanged.}

\subsection{Policy-Level Latency--Assurance Trade-off}
\label{subsec:policy_tradeoff}

\begin{table*}
\centering
\begingroup
\color{blue}
\caption{E1 verification latency for valid
\texttt{submit\_claim} bundles. Each profile summarizes
300 case medians, each based on three invocations.
P95 is calculated across these case medians.}
\label{tab:profile_tradeoff}
\footnotesize
\renewcommand{\arraystretch}{1.10}
\setlength{\tabcolsep}{5pt}

\begin{tabular*}{0.96\textwidth}
{@{\extracolsep{\fill}} l l r r r @{}}
\toprule
\textbf{Profile} &
\textbf{Profile change} &
\textbf{\shortstack{Median [95\% CI]\\(ms)}} &
\textbf{\shortstack{P95\\(ms)}} &
\textbf{\shortstack{Reduction vs.\\\texttt{ALL\_ON}}} \\
\midrule
\texttt{ALL\_ON} &
None &
929.94 [925.16--936.15] & 1065.59 & -- \\

\texttt{NO\_WIT} &
Bundle-witness requirement &
934.06 [929.41--939.42] & 1071.89 & $-0.44\%$ \\

\texttt{NO\_AOEP} &
AOEP continuity &
929.50 [926.98--932.63] & 1057.01 & $0.05\%$ \\

\texttt{NO\_CAP} &
Capability-chain validation &
931.47 [928.36--935.87] & 1066.11 & $-0.16\%$ \\

\texttt{NO\_EPOCH} &
Checkpoint-context validation &
351.80 [347.12--356.28] & 451.77 & $62.17\%$ \\
\bottomrule
\end{tabular*}

\par\smallskip
\begin{minipage}{0.96\textwidth}
\scriptsize
\emph{Note:}
Reductions compare profile medians with \texttt{ALL\_ON}.
Negative values indicate higher medians.
\end{minipage}
\endgroup
\end{table*}

\textcolor{blue}{Experiment E1 compares verifier-side latency across five policy profiles for the \texttt{submit\_claim} event. This event combines prior-event verification with capability, witness, continuity, and checkpoint-context evidence.}

\textcolor{blue}{Table~\ref{tab:profile_tradeoff} reports the counterbalanced profile comparison. All 900 timed invocations per profile returned \textsc{Accept}. The \texttt{NO\_EPOCH} median was 351.80~ms, compared with 929.94~ms under \texttt{ALL\_ON}, giving the largest observed conditional reduction of 62.17\%. The paired comparisons gave a Holm-adjusted
\(p=2.43\times10^{-50}\) for \texttt{NO\_EPOCH}
and \(p\geq0.185\) for the other three profiles.}

\textcolor{blue}{These profiles change specific evidence requirements. Under \texttt{NO\_WIT}, the bundle witness is optional, but any supplied witness signature is verified, and WCL signature checks remain active. Under \texttt{NO\_CAP}, capability-chain validation is disabled, while \texttt{submit\_claim} still requires verification of its prior \texttt{send\_result} bundle. \texttt{NO\_AOEP} disables append-only continuity checking, which otherwise applies when a previous checkpoint is supplied, and \texttt{NO\_EPOCH} disables WCL integrity and target-checkpoint membership checks. The measured differences are conditional comparisons with \texttt{ALL\_ON} and should not be interpreted as isolated costs of individual checks.}

\textcolor{blue}{Within the revised E1 and E2 study, the independently timed
\texttt{ALL\_ON} claim medians were 929.94~ms in the standalone condition
and 927.21~ms in the four-event sequence. The median of the 300 paired
sequence-minus-standalone case differences was $-2.2644$~ms, with a 95\%
case-bootstrap interval of $[-6.4574,\,4.3811]$~ms. The interval does not
identify a clear nonzero median context shift within this execution;
it does not establish equivalence.}

\textcolor{blue}{A full-protocol repeat used the same frozen corpus,
verifier, profiles, schedule and timing procedure, with a new valid-input
gate, 27 excluded warm-ups and 8100 further timed calls. The first completed
study remains the principal analysis. Repeat medians were 1.36--2.90\%
lower across the nine conditions. Its paired claim-context estimate was
$-2.5866$~ms, with a 95\% interval of $[-7.70795,\,1.7457]$~ms. The same
300 cases were reused; the two executions are reported separately rather
than pooled. Complete repeat summaries and individual-call tails are
retained in S1.}

\subsubsection{\textcolor{blue}{Experiment E4: Unified Timing Comparison}}
\label{subsec:unified_timing}

\textcolor{blue}{A separate revision study used the unchanged historical verifier, before the signed-reference correction. Thirty complete cases from the original E2 dataset (ranks 1, 11, \ldots, 291) supplied identical bundles, prerequisite evidence, and public verification material in both conditions. Both used \texttt{ALL\_ON} and the same frozen 1200-record WCL. The runtime was Python 3.9.13 with PyNaCl 1.6.2 on Windows build 26300 and an Intel Core i7-13700H processor.}

\textcolor{blue}{Condition A verified only \texttt{submit\_claim} bundles. Condition B verified each case's four events in workflow order, timing the claim separately. Blocks A1, B1, B2, A2 used the same ascending case order, giving two claim measurements per condition per case. Before each block, a separate case supplied one warm-up call for A or four for B. Each timed call used a fresh verifier subprocess; the interval included startup, imports, evidence access, verification, output capture, and process termination. Bytecode writes were disabled. OS caches were not flushed, and background load and power behavior were not controlled. All 124 preflight calls, 300 measured calls, and 10 excluded warm-up calls accepted. Table~\ref{tab:unified_timing} summarizes the 120 measured claim calls.}

\begin{table}[!t]
\centering
\begingroup
\color{blue}

\caption{\textcolor{blue}{Unified timing comparison of \texttt{submit\_claim} latency for 30 cases, with two measurements per case per condition. Each row summarizes 60 claim calls.}}

\label{tab:unified_timing}
\footnotesize
\renewcommand{\arraystretch}{1.10}
\setlength{\tabcolsep}{3pt}
\begin{tabularx}{\columnwidth}{@{}Xrr@{}}
\toprule
\textbf{Condition} &
\textbf{Median (ms)} &
\textbf{P95 (ms)} \\
\midrule
A: claims only & 926.43 & 1098.03 \\
B: four-event sequence & 925.80 & 1062.31 \\
\bottomrule
\end{tabularx}
\endgroup
\end{table}

\textcolor{blue}{For each case, we calculated the median of its two timings per condition and then the B-minus-A difference. The median of these 30 paired differences was $-15.22$~ms, with a 95\% percentile case-bootstrap interval of $[-44.60,\,8.36]$~ms (3000 resamples; seed 42). This paired estimate differs from subtracting the pooled medians in the table. Both conditions were faster in their second block; block claim medians ranged from 874.30 to 977.54~ms. The interval describes case-level variation within this execution and does not capture between-block or between-session uncertainty. These results establish neither equivalence nor the cause of the historical E1--E2 discrepancy discussed in Section~\ref{subsec:limitations_validity}.}

\subsection{\textcolor{blue}{Experiment E5:
Checkpoint-Context Stage Profiling}}
\label{subsec:wcl_stage_profile}

\textcolor{blue}{A separate revision study profiled checkpoint-context verification under \texttt{ALL\_ON}, using ten preserved \texttt{submit\_claim}/\texttt{send\_result} pairs and WCL lengths of 80, 400, and 1200 records. The verifier included sender-authentication and signed-message-field corrections, before the version~1 signed-reference correction. The run used an Intel Core i7-13700H, Windows build 10.0.26300, Python 3.9.13, and PyNaCl 1.6.2. Each case, WCL length, and verifier variant received one warm-up followed by five fresh-process measurements. Baseline and instrumented variants used randomized matched ordering with seed 42 and produced matching outcomes for 30 valid configurations and eight negative controls. OS caches were not flushed; background load and power mode were uncontrolled.} \textcolor{blue}{Table~\ref{tab:wcl_stage_profile} reports the mean time spent in each WCL processing stage.}

\begin{table}[!t]
\centering
\begingroup
\color{blue}
\caption{\textcolor{blue}{Mean WCL-stage time per claim. Values are in milliseconds; column headings give WCL record counts.}}
\label{tab:wcl_stage_profile}
\footnotesize
\renewcommand{\arraystretch}{1.08}
\setlength{\tabcolsep}{3pt}
\begin{tabularx}{\columnwidth}{@{}Xrrr@{}}
\toprule
\textbf{Stage} & \textbf{80} & \textbf{400} & \textbf{1200} \\
\midrule
Buffered file access
& 1.15 & 2.74 & 6.71 \\
JSON parsing
& 1.13 & 5.28 & 15.04 \\
Canonicalization and hashing
& 2.52 & 11.76 & 34.17 \\
Key lookup and construction
& 0.43 & 1.98 & 5.68 \\
Signature-input decoding
& 0.74 & 3.36 & 9.85 \\
Record-signature verification
& 17.14 & 82.83 & 246.21 \\
Checkpoint witness-signature verification
& 17.12 & 82.94 & 246.55 \\
Target-checkpoint comparison
& 0.03 & 0.14 & 0.41 \\
Other checks and bookkeeping
& 1.43 & 6.48 & 18.76 \\
\midrule
\textbf{Total WCL time}
& \textbf{41.69} & \textbf{197.51} & \textbf{583.36} \\
\bottomrule
\end{tabularx}

\par\smallskip
\begin{minipage}{\columnwidth}
\scriptsize
\emph{Notes:}
Arithmetic means of 50 instrumented claim calls per WCL length
(ten preserved pairs, five repetitions), combining both WCL scans.
Instrumentation costs are retained.
Buffered access includes Python buffering and text decoding.
Other checks and bookkeeping form the residual interval.
Rounding may affect sums.
\end{minipage}
\endgroup
\end{table}

\textcolor{blue}{The measured WCL interval excludes process startup and report-file
writing. Each accepted claim performed two complete WCL scans, covering
its own bundle and its prerequisite bundle. Each scan verified two
signatures per record and continued to the end after locating the target
checkpoint. Signature verification accounted for 82.2--84.5\% of mean WCL
time. Longer retained WCLs increased the measured mean processing time.
Here, history length denotes the number of retained WCL records scanned.
These internal stage estimates characterize the instrumented E5 verifier;
they are not an additive breakdown of the end-to-end subprocess
latencies measured in E1 and E2.}

\subsection{Workflow-Event Evidence Requirements}
\label{subsec:event_eval}

\begin{figure*}

\centering
\includegraphics[width=0.77\textwidth]{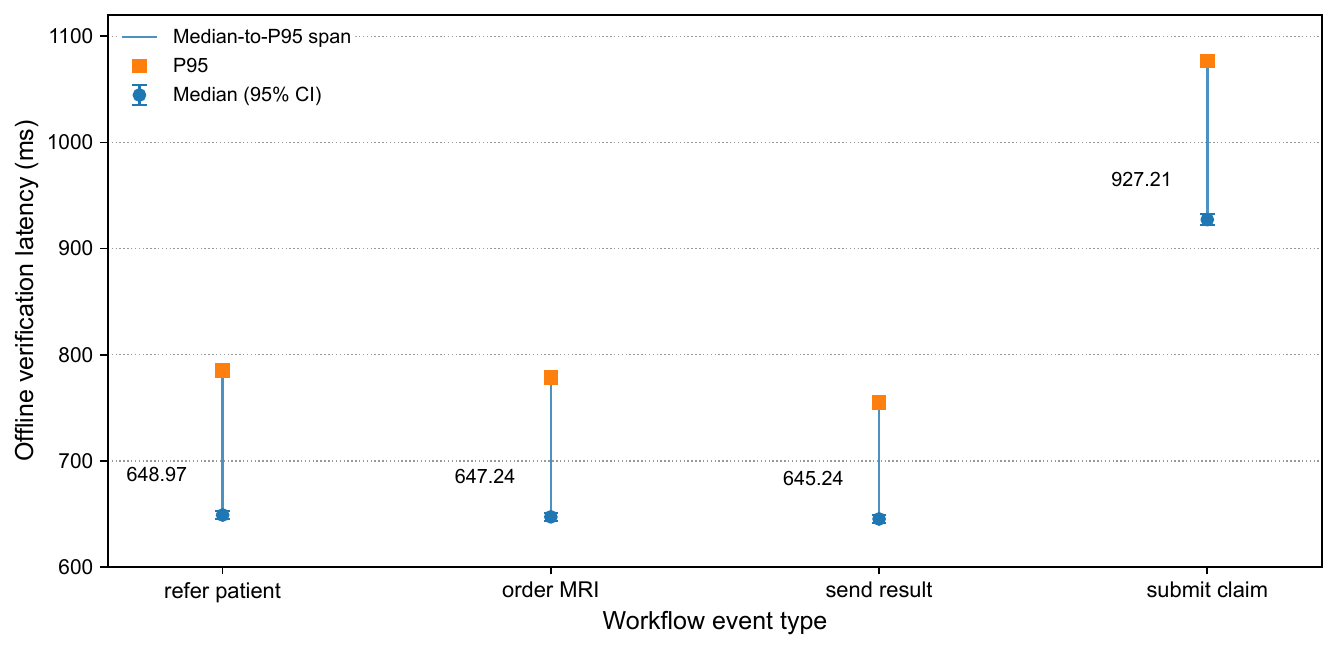}

\caption{\textcolor{blue}{E2 verification latency under \texttt{ALL\_ON}.
Each event type summarizes 300 paired case medians, each based on three
invocations. Circles show the median with its 95\% case-bootstrap interval,
squares show P95 across case medians, and vertical spans connect the median
to P95. The y-axis is truncated for readability.}}

\label{fig:event_latency}
\end{figure*}

\textcolor{blue}{Experiment E2 compares verifier-side latency across workflow
event types under \texttt{ALL\_ON}. Figure~\ref{fig:event_latency} shows
the median and 95\% case-bootstrap interval for each type; P95 of case
medians is reported below. This comparison concerns event-specific
evidence requirements rather than workflow throughput. The verifier
checks message signatures together with applicable prerequisite,
capability, continuity, witness and checkpoint-context evidence.}

\textcolor{blue}{Figure~\ref{fig:event_latency} reports median case-median
latencies of 648.97~ms for \texttt{refer\_patient}, 647.24~ms for
\texttt{order\_mri}, 645.24~ms for \texttt{send\_result}, and 927.21~ms
for \texttt{submit\_claim}. P95 across case medians was 1076.75~ms for
the claim and 755.17--784.94~ms for the other event types. The Friedman
test gave $\chi^2=539.772$ (3 degrees of freedom),
$p=1.1451\times10^{-116}$ and Kendall's $W=0.599747$.
Holm-adjusted paired Wilcoxon comparisons separated the claim from each
other event type ($p=3.65\times10^{-50}$ for each comparison).
Among the other events, \texttt{refer\_patient} versus
\texttt{send\_result} had adjusted $p=0.01967$; the other two comparisons
had $p=0.94857$ and $p=0.08894$. In the separate full-protocol repeat,
the \texttt{refer\_patient}--\texttt{send\_result} comparison had adjusted \mbox{$p=0.30101$}.
This contrast reached the 0.05 threshold after Holm
adjustment in the principal run, but not in the repeat.}

\textcolor{blue}{The E2 verifier checks a claim's capability chain, recursively
verifies its separately preserved \texttt{send\_result} prerequisite,
and compares the authenticated prior evidence with the claim's signed
reference. Each claim in this corpus contains two capability tokens,
whereas each other event bundle contains one. Under \texttt{ALL\_ON},
claim and prerequisite verification each trigger a complete WCL scan,
including record-signature and checkpoint witness-signature checks for
every record. This path adds file access, parsing, hashing and
signature-verification work. Section~\ref{subsec:wcl_stage_profile}
profiles WCL processing in a separate earlier instrumented verifier;
Table~\ref{tab:e2_bundle_sizes} reports sizes for the revised E2 corpus.
These observations identify the additional processing path rather than
isolating the latency cost of its individual operations.}

\textcolor{blue}{Table~\ref{tab:e2_bundle_sizes} reports storage statistics
for the same 1200 bundles used in the revised E1 and E2 measurements.
Sizes are uncompressed JSON file lengths in bytes, including embedded
evidence. P95 uses linear interpolation. Separately stored prerequisite bundles and the shared WCL are excluded from each individual bundle size.}

\begin{table}[!t]
\centering
\begingroup
\color{blue}
\caption{Serialized sizes in the E2 corpus
(bytes; $n=300$ bundles per event type).}
\label{tab:e2_bundle_sizes}
\footnotesize
\renewcommand{\arraystretch}{1.10}
\setlength{\tabcolsep}{3pt}
\begin{tabular}{@{}lrrr@{}}
\toprule
\textbf{Event} & \textbf{Median} & \textbf{P95} & \textbf{Min--max} \\
\midrule
\texttt{refer\_patient} & 4,594 & 4,810 & 3,718--5,031 \\
\texttt{order\_mri} & 4,909 & 5,125 & 4,467--5,346 \\
\texttt{send\_result} & 5,123 & 5,339 & 4,681--5,560 \\
\texttt{submit\_claim} & 6,928 & 7,144 & 6,477--7,370 \\
\bottomrule
\end{tabular}
\endgroup
\end{table}

\textcolor{blue}{A \texttt{submit\_claim} bundle together with its required
\texttt{send\_result} prerequisite has a median combined size of 12,051 bytes, a P95 of 12,483 bytes, and a range of 11,158--12,930 bytes. These values exclude the shared 1200-record WCL, whose serialized size is 1,290,183 bytes. Prerequisite files are already among the 1200 corpus bundles and are counted once in corpus storage. These measurements show a kilobyte-scale per-event footprint with shared checkpoint history; they do not isolate the contribution of storage size to latency.}

\subsection{Negative-Evidence Diagnostic Coverage}
\label{subsec:tamper_eval}
Experiment E3 focuses on rejection behavior, not latency comparison. Each case introduces one targeted evidence corruption or workflow misuse operation.

Table~\ref{tab:negative_evidence_diagnostic} maps each case to the verifier check expected to fail first, the corresponding evidence-level security threat, and the verifier-supported property.

Each negative-evidence case was generated from an otherwise valid evidence bundle by applying one targeted mutation or misuse operation. The cases are not intended to be mutually exclusive attack classes. Each case is designed to exercise a specific verifier predicate and confirm rejection at the expected first-failing check.

\begin{table*}
\centering
\caption{Negative-evidence cases mapped to evidence-level security threats and verifier-supported properties.}
\label{tab:negative_evidence_diagnostic}
\footnotesize
\renewcommand{\arraystretch}{1.08}
\setlength{\tabcolsep}{4pt}
\begin{tabularx}{\textwidth}{@{}C{0.05\textwidth} L{0.22\textwidth} C{0.11\textwidth} C{0.09\textwidth} X@{}}
\toprule
\textbf{Case} & \textbf{Verifier check} & \textbf{Threat ID} & \textbf{Property} & \textbf{Mutation or misuse} \\
\midrule
C1  & Payload binding & T3 & P2 & Canonical payload modified after issuance. \\
C2  & Merkle inclusion & T3 & P1 & Audit-path sibling hash replaced. \\
C3  & Merkle inclusion & T3 & P1 & Leaf index modified while preserving the bundle structure. \\
C4  & Signed checkpoint validation & T2 & P1 & Signed checkpoint root modified. \\
C5  & Append-only continuity & T5 & P3 & AOEP proof element modified between checkpoints. \\
C6  & Capability binding & T6 & P5 & Capability scope modified. \\
C7  & Capability binding & T6 & P5 & Capability evidence rebound to another message commitment. \\
C8  & Witness validation & T5 & P4 & Witness signature over checkpoint material modified. \\
C9  & Workflow prerequisite omission & T6 & P5 & Prior bundle for \texttt{submit\_claim} removed. \\
C10 & \textcolor{blue}{Workflow prerequisite action check} & T6 & P5 & \textcolor{blue}{Prior bundle's unsigned action label changed from \texttt{send\_result} to \texttt{order\_mri}; the signed message remained unchanged.} \\
C11 & \textcolor{blue}{Workflow prerequisite case check} & T6 & P5 & \textcolor{blue}{Case identifier in the prior bundle's unsigned metadata changed to create a mismatch; its signed message remained unchanged.} \\
C12 & \textcolor{blue}{Workflow prerequisite examination check} & T6 & P5 & \textcolor{blue}{Examination identifier in the prior bundle's unsigned metadata changed to create a mismatch; its signed message remained unchanged.} \\
C13 & Witness-evidence omission & T5 & P4 & Required witness evidence removed. \\
C14 & \textcolor{blue}{Checkpoint-context availability} & T5 & P6 & \textcolor{blue}{Required WCL file temporarily moved from its expected path, making it unavailable during verification.} \\

C15 & Signed checkpoint validation & T2 & P1 & Previous checkpoint signature modified. \\
C16 & Append-only continuity & T5 & P3 & AOEP extension-size relation modified. \\
C17 & Sender authentication & T1 & P2 & Sender signature over canonical message corrupted. \\
\bottomrule
\end{tabularx}

\vspace{0.25em}
\begin{minipage}{0.98\textwidth}
\scriptsize
\emph{Note:} \(T_i\) refers to the evidence-level security threat in Table~\ref{tab:evidence_failure_mapping}, and \(P_i\) refers to the verifier-supported property in Table~\ref{tab:verifier_properties}.
\end{minipage}
\end{table*}

\textcolor{blue}{Under the fully enabled profile, all 17 cases in the original E3 diagnostic suite were rejected. These results confirm the targeted rejection behavior exercised by this suite across message and checkpoint evidence, delegation, workflow prerequisites, and checkpoint context. The reported first-failing checks identify why the tested bundles failed that profile. This diagnostic result is specific to the original verifier, test cases, and profile. It does not establish rejection of all adversarial combinations of individually valid evidence.}

\paragraph{\textcolor{blue}{Experiment E6: Valid Evidence in an Incorrect Context}}

\textcolor{blue}{Additional functional tests used unchanged, independently valid donor evidence. They were separate from the original 17 cases and used a frozen verifier with the sender-authentication correction, before the signed-canonical-field and prior-reference corrections. Donors passed correct-context controls before substitution. Ten contexts in each of five classes were tested under all five profiles: a result bundle with a different workflow case identifier, a result bundle with the same case identifier but a different examination identifier, a capability chain for another message, witness evidence for another checkpoint, and a donor logging/checkpoint package combined with the recipient's signed message. Of 250 substitution observations, 240 were rejected; the ten capability substitutions under \texttt{NO\_CAP} were accepted as specified before execution. This profile omits capability-chain verification. Acceptance still requires the mandatory message and prerequisite checks and does not establish delegated authority. The logging-package condition first failed payload binding and therefore does not isolate Merkle-proof transplantation.}

\textcolor{blue}{An exploratory condition kept the signed claim's reference to result A but resolved valid result B, with the same case, examination, and action, at A's lookup path. All 50 observations were accepted. The donor occurred later in the synthetic sequence. A separately specified follow-up used an earlier valid result and again observed 50 acceptances. The earlier verifier checked the prerequisite's validity, type, case, and examination, but did not authenticate the association between the lookup identifier and the resolved evidence. The issuer instead required the particular result recorded in its workflow state. These observations identified the discrepancy addressed by the version 1 reference below; they are retained as results of the earlier verifier.}

\textcolor{blue}{The primary dataset contains 950 top-level verifier calls (excluding recursive prerequisite calls): 650 accepted controls, 250 specified substitutions, and 50 exploratory substitutions. The separate follow-up contains 100 accepted control rechecks and 50 exploratory substitutions. 
\textcolor{blue}{The primary study reused ten preserved
donor--recipient pairs and ten newly generated synthetic
contexts across their applicable classes and profiles;
the follow-up reused the latter contexts.}
No unexpected exception, worker timeout, or nonzero worker exit was recorded. These are functional observations, separate from latency measurements.}

\paragraph{\textcolor{blue}{Experiment E7:
Signed Prior-Evidence Reference Verification}}

\textcolor{blue}{A separate revision study compared a baseline that already used signed canonical fields for workflow decisions with the version 1 correction. In the corrected \texttt{submit\_claim}, the signed \texttt{prior\_evidence\_ref} commits to the prior message hash, sender key, leaf position, and checkpoint identity: log key, root, size, and signed checkpoint creation timestamp. The verifier compares this tuple with the validated prior evidence under every profile, including \texttt{NO\_CAP}; the event identifier remains a lookup locator. An independent offline factory generated 114 public signed fixtures across ten synthetic contexts, including 24 deliberately malformed reference or locator variants. Expectations were fixed before execution. The recorded runtime was Python 3.9.13, PyNaCl 1.6.2, and Windows build 10.0.26300.}

\begin{table*}[t]
\centering
\color{blue}

\caption{\textcolor{blue}{Version~1 prior-reference study.}}
\label{tab:prior_reference_results}
\footnotesize
\renewcommand{\arraystretch}{1.08}

\begin{tabularx}{\textwidth}
{@{}X L{0.22\textwidth} L{0.22\textwidth}@{}}
\toprule
\textbf{Condition} &
\textbf{Canonical-field baseline} &
\textbf{Version 1 correction} \\
\midrule

\multicolumn{3}{@{}l}{
\textbf{A. Reference-binding tests and valid controls}} \\

Correct-context controls
& 350 Accept & 350 Accept \\

Earlier valid result, same case and examination
& 50 Accept & 50 Reject \\

Later valid result, same case and examination
& 50 Accept & 50 Reject \\

Same canonical message, different logged occurrence
& 50 Accept & 50 Reject \\

Wrong valid donor with changed unsigned labels
& 50 Accept & 50 Reject \\

Same committed evidence through an alias locator
& 50 Accept & 50 Accept \\

\midrule
\multicolumn{3}{@{}l}{
\textbf{B. Legacy-format compatibility}} \\

Legacy claim without the signed reference
& 50 Accept & 50 Reject \\

\midrule
\multicolumn{3}{@{}l}{
\textbf{C. Corrected-verifier-only tests}} \\

Signed reference or locator variants (24)
& --- & 120 Reject \\

Regression classes (10)
& --- & 40 expected Accept; 460 Reject \\

\midrule
\textbf{Total}
& \textbf{650 Accept} (650 calls)
& \textbf{440 Accept; 830 Reject} (1,270 calls) \\
\bottomrule
\end{tabularx}

\par\smallskip
\begin{minipage}{\textwidth}
\scriptsize
\emph{Notes:}
Counts are top-level verifier calls; internal recursive checks
are excluded.
The 350 controls per version reuse 70 bundles from ten contexts
under five profiles; alias controls are additional.
Each 50-call row uses ten contexts under five profiles per version.
\textcolor{black}{\textcolor{blue}{A dash} identifies conditions
executed only with the corrected verifier.}
\end{minipage}
\end{table*}

\textcolor{blue}{Table~\ref{tab:prior_reference_results} separates the versions and conditions. Each wrong-reference class used ten contexts under each profile, and the corrected verifier rejected all 200 wrong-reference observations with \texttt{prior\_reference\_mismatch}. Correct-context controls were accepted before substitution in all 700 observations across the two versions. \textcolor{black}{\textcolor{blue}{Each version} accepted \textcolor{blue}{all} 50 alias controls, which resolved to the same committed evidence.} Legacy claims lacking the reference were deliberately rejected by the correction rather than silently upgraded. Each of the 24 signed schema or locator variants contributed one fixture under each of the five profiles, giving 120 controlled rejections.}

\textcolor{blue}{The ten regression classes concerned payload binding, sender signatures, checkpoint signatures, inclusion proofs, invalid present witness signatures, missing witness evidence, missing extension evidence, missing capabilities, unavailable WCL, and an unsigned action change intended to bypass the claim check. Each class used ten contexts under five profiles. The 40 prespecified acceptances were ten each for witness omission under \texttt{NO\_WIT}, extension omission under \texttt{NO\_AOEP}, capability omission under \texttt{NO\_CAP}, and WCL unavailability under \texttt{NO\_EPOCH}; all other 460 regression observations were rejected. An invalid present witness signature caused rejection even under \texttt{NO\_WIT}. Across both versions and all phases, all 1920 top-level verifier calls matched their fixed expectations. No unexpected exception, worker timeout, nonzero worker exit, or inconclusive outcome was recorded. Expected semantic rejection exceptions were classified as controlled rejections.}

\textcolor{blue}{Schema and regression conditions varied one specified field or omission at a time. The donor-relabelling condition combined valid evidence substitution with changes to unsigned labels. No randomized multi-mutation campaign was performed. Contexts were reused across classes and profiles, so the counts do not represent independent workflows or an estimated false-acceptance rate.}

\textcolor{blue}{A separate set of 75 integration observations also matched fixed expectations: 41 issuer checks, nine constructor checks, and 25 actual evidence-loader checks. Issuer persistence, log append, capability validation, and WCL emission used documented test doubles. Constructor HTTP transport, key loading, and evidence loading were replaced by documented test substitutes, while the separate actual-loader checks used the real verifier and public fixtures. These were not live workflow runs. The results support agreement with the committed prior evidence under the tested conditions; they do not establish global latestness or general parser robustness. All observations remain separate from the historical E3 suite, the earlier context studies, and all latency datasets.}

\paragraph{\textcolor{blue}{Experiment E8:
Input-Corruption Handling}}
\textcolor{blue}{A separate functional input-corruption study
used two existing synthetic contexts, targeting each claim
and its recursively verified prerequisite. Twenty variants
per target covered syntax or encoding (10), structure or
field types (5), and mandatory authenticated inputs (5),
giving 400 top-level negative-input calls across five profiles,
plus 70 controls comprising ten unchanged-claim calls and
60 equivalent-JSON representation calls. Inputs and expected
outcomes were fixed before execution. Before the input-handling
repair, the 470 calls produced 70 accepted controls,
120 explicit rejections, and 280 escaped exceptions
(200 syntax or encoding exceptions and 80 structure or type
exceptions), which the test harness contained. A narrow repair
retained the existing cryptographic, signed-reference, and
policy-check code. Replaying the identical frozen corpus
accepted all 70 control calls and explicitly rejected all
400 negative-input calls, with no escaped exception,
inconclusive return, worker timeout, or nonzero worker exit.
Each negative condition applied one corruption operation,
and contexts were reused across conditions and profiles.
These bounded functional results are separate from
Table~\ref{tab:prior_reference_results}, the original E1--E3
experiments, and all latency measurements; the calls are
not independent workflows and do not establish general
parser robustness.}

\subsection{\textcolor{blue}{Experiments E9a--E9d:
Sensitivity to Evidence and Workflow Size}}
\label{subsec:sensitivity_eval}

\textcolor{blue}{Four separate revision studies examined retained checkpoint
history/WCL length, signed-message size, claim delegation-chain length,
and workflow-event count under \texttt{ALL\_ON}. They used the frozen
verifier containing the sender-authentication and signed-message-field
corrections, before the version~1 signed-reference correction. Their
measurements are separate from the revised E1 and E2 measurements and
the historical E3 diagnostics.}

\textcolor{blue}{Each condition used ten cases, or ten matched case families in the workflow-length study, with one warm-up trial followed by five measured repetitions per case and condition. Table~\ref{tab:sensitivity_results} reports the median of the ten case or family medians. Claim timing covered a fresh verifier subprocess from launch through captured completion, including startup, imports, evidence access, verification, output, and exit. Event-batch timing covered sequential verification of all event bundles using one fresh subprocess per event. Condition order was randomized with seed 42.}

\textcolor{blue}{The recorded software environment was Windows build 10.0.26300, Python 3.9.13, and PyNaCl 1.6.2. The message-size, delegation, and workflow studies additionally recorded an Intel Core i7-13700H; the WCL-run record did not capture the CPU model. OS caches were not flushed, and garbage collection, background load, and CPU power state were not explicitly controlled.}

\textcolor{blue}{The WCL study kept the preserved claim/prerequisite bundles fixed while extending journal prefixes from 80 to 400 and 1200 records. The additional records occurred after the fixed target checkpoints. The other studies used offline evidence fixtures with 80-record WCLs. Message-size conditions varied the complete signed canonical JSON length, with matched prerequisite evidence and proof dimensions. Delegation conditions varied the claim's capability-token count while keeping its signed message and WCL unchanged; the prerequisite retained two capability tokens. Parenthesized bundle sizes in the table show the accompanying serialized evidence growth.}

\textcolor{blue}{Workflow cases contained one referral, one, two, or four order/result pairs, and one final claim, giving 4, 6, or 10 events. Within each matched family, the events used the same final 80-leaf checkpoint and 80-record WCL, with fixed message lengths, capability-chain lengths, and proof dimensions. A final claim verified its designated last result. Event-batch measurements instead checked all event bundles individually in sequence.}

\begin{table*}[!t]
\centering
\begingroup
\color{blue}

\caption{\textcolor{blue}{Separate sensitivity studies. Latency values are in milliseconds.}}
\label{tab:sensitivity_results}
\footnotesize
\renewcommand{\arraystretch}{1.12}
\setlength{\tabcolsep}{4pt}
\begin{tabularx}{\textwidth}
{@{}L{0.24\textwidth} L{0.39\textwidth}
>{\centering\arraybackslash}X
>{\centering\arraybackslash}X@{}}
\toprule
\textbf{Varied factor} &
\textbf{Condition} &
\textbf{Final claim (ms)} &
\textbf{Event batch (ms)} \\
\midrule
\textcolor{blue}{E9a: Retained history/WCL} &
80 records (83,670 WCL bytes) & 299.19 & -- \\
&
400 records (418,690 WCL bytes) & 468.78 & -- \\
&
1200 records (1,256,490 WCL bytes) & 852.66 & -- \\
\midrule
\textcolor{blue}{E9b: Signed canonical JSON length} &
1 KiB (5,256 bundle bytes) & 303.82 & -- \\
&
16 KiB (25,736 bundle bytes) & 315.03 & -- \\
&
64 KiB (91,272 bundle bytes) & 307.93 & -- \\
\midrule
\textcolor{blue}{E9c: Claim delegation chain} &
1 token (4,922 bundle bytes) & 322.72 & -- \\
&
2 tokens (5,555 bundle bytes) & 344.61 & -- \\
&
4 tokens (6,821 bundle bytes) & 331.22 & -- \\
\midrule
\textcolor{blue}{E9d: Workflow length} &
4 events & 320.37 & 1214.47 \\
&
6 events & 315.44 & 1784.99 \\
&
10 events & 322.50 & 2969.10 \\
\bottomrule
\end{tabularx}

\par\smallskip
\begin{minipage}{0.98\textwidth}
\scriptsize
\emph{Notes:}
Latency entries are medians of ten case or family medians,
each based on five measured repetitions after warm-up.
Parenthesized bundle sizes are median serialized claim-file sizes.
One KiB equals 1024 bytes.
A dash indicates that event-batch timing was not measured.
Each block represents a separate study; comparisons are within blocks.
\end{minipage}
\endgroup
\end{table*}

The retained-history/WCL and event-batch medians increased across the tested levels. \textcolor{blue}{Final-claim latency} did not \textcolor{blue}{increase monotonically with signed-message size, claim-chain length, or workflow length across the tested levels}. These are descriptive comparisons within the tested ranges and process-level timing boundary. In the WCL study, retained checkpoint-record count, journal length, and journal bytes varied together, while target checkpoints, tree sizes, and extension-proof spans remained fixed. The study therefore characterizes a coupled history/WCL axis. Its process-level latencies are distinct from the internal WCL-stage means in Section~\ref{subsec:wcl_stage_profile}.

\subsection{\textcolor{blue}{Experiment E10: Receiver Admission Receipts}}
\label{subsec:receipt_eval}

\textcolor{blue}{E10 evaluated receiver-side admission-receipt issuance and offline verification in a candidate extending the input-handling-corrected implementation from E8. Receiver-key authorization and admission-policy identity and version were trusted external inputs. Each of the five existing profiles was tested with an independent optional or required receipt setting. Missing receipts were permitted only in optional mode; supplied receipts were always verified, including those of recursively checked prerequisites. Inputs and expected outcomes were fixed before execution.}

\textcolor{blue}{Four correct-context bundles per synthetic context---an order, a claim, and two result occurrences---were checked across ten existing contexts, five profiles, and both receipt modes, giving 400 control calls. All controls were accepted before the mutation phase. Sixteen mutation classes covered missing evidence, signature and key failures, message/event/occurrence binding, receiver/policy/state binding, and receipt-format errors, giving 1600 further calls. Overall, the 2000 calls produced 500 accepts and 1500 explicit rejections, matching their fixed expectations. The 100 acceptances beyond the controls involved missing root or prerequisite receipts exclusively in optional mode, with 50 of each. All 700 validly signed binding/state mutation calls reached their designated rejection reasons.}

\textcolor{blue}{Separate file-backed integration tested receipt issuance and compatibility with base-format evidence. It used ten fresh four-event workflows with the receipt-enabled issuer and ten with the base issuer. The 140 issuer calls produced 40 outputs containing admission receipts, 40 base-format outputs without admission receipts, and 60 controlled refusals. Six predefined failure conditions per new workflow exercised sender, capability, prerequisite, signed-reference, receiver, and key checks. Refusals produced no admission signature or new admission receipt. Ten workflow-prerequisite refusals nevertheless appended a log leaf and checkpoint before rejection. Verification used 450 calls: 200 receipt-required accepts for new outputs, 200 optional-mode accepts for base outputs, and 50 required-missing rejections for base claims. Each synthetic workflow used one configured receiver gateway.}

\textcolor{blue}{With the new receipt requirement disabled, separate replays retained all fixed expectations of the prior-reference suite (1920 calls), including the weaker baseline's expected wrong-reference acceptances, and the input-handling suite (470 calls). \textcolor{black}{\textcolor{blue}{The reported counts} exclude internal recursive \textcolor{blue}{verifier} calls\textcolor{blue}{. Synthetic contexts were reused.}} \textcolor{blue}{The tested receipts authenticate receiver-admission assertions bound to the specified message, logged occurrence, receiver, and admission policy. Log inclusion alone does not establish receiver admission.} These functional results remain separate from all latency measurements.}

\section{Discussion}
\label{sec:discussion}

\subsection{Interpretation of Results} \label{subsec:discussion_interpretation} 

\textcolor{blue}{The evaluated prototype supports policy-relative review of preserved workflow-event evidence. When required by the selected profile and successfully verified, sender signatures, log evidence, and capability chains support claims of authorship, authenticated log commitment, and delegated authority, respectively.}

\textcolor{blue}{The separate profiling study in Section~\ref{subsec:wcl_stage_profile} quantifies the stages of checkpoint-context verification and identifies signature verification across complete WCL scans as the largest contributor to measured WCL-processing time. These measurements complement the conditional policy-profile comparisons in E1.}

\textcolor{blue}{Experiment E2 reports higher latency for \texttt{submit\_claim}, whose verification includes a separately verified prerequisite and additional capability checks. Section~\ref{subsec:event_eval} explains this processing path alongside the bundle-size statistics.}

\textcolor{blue}{Experiment E3 reports the first-failing checks for the rejected test bundles, helping explain evidence insufficiency under the tested profile. Table~\ref{tab:negative_evidence_diagnostic} maps the cases to their verification checks, threats, and verifier-supported properties.}

\textcolor{blue}{The following examples illustrate the model's evidence requirements for different disputes. In an authorship or repudiation dispute, the relevant evidence would normally be sender-authentication evidence, payload binding, and authenticated log commitment. In a receipt dispute, the verifier must require an explicit receiver-signed receipt object, because receipt is not supported by log inclusion, storage, or transport behavior alone. In a delegation dispute, the relevant evidence is the capability chain and any policy-required workflow evidence showing that the actor was authorized to act on behalf of another actor. In a workflow-prerequisite dispute, the verifier checks the preserved prior-event evidence and its binding to the current bundle. For example, a dispute may concern whether a \texttt{submit\_claim} was supported by a prior \texttt{send\_result}. If the dispute concerns checkpoint continuity or checkpoint context, the profile should require the corresponding AOEP, witness-backed checkpoint evidence, or WCL/checkpoint-context material.}

\textcolor{blue}{Policy selection should follow the dispute's evidence requirements rather than latency alone. Acceptance supports only claims whose required checks are included in the selected profile and succeed.}

\subsection{Limitations and Threats to Validity}
\label{subsec:limitations_validity}

Several limitations qualify the interpretation of the results. First, the reported timings are verifier-side prototype measurements rather than production measurements.

\textcolor{blue}{They describe latency in the evaluated workloads and policy profiles; their absolute values depend on the implementation and execution conditions.}

Larger deployments, multi-host execution, concurrent agents, different storage backends, and longer checkpoint histories may change the absolute latency. The current evaluation also does not measure operational deployment factors such as network delay, organizational key-management processes, or independent witness availability.

\textcolor{blue}{A second limitation is that the case-bootstrap intervals
resample cases within one execution; they do not capture between-session
uncertainty. The full-protocol repeat shows variation in absolute
latencies and individual-call tails, even with the same frozen inputs
and verifier. Its results remain separate from the principal analysis.}

\textcolor{blue}{The historical E1 and E2 medians, 785.52~ms and
680.21~ms, came from separate executions with disjoint claim-event sets.
Their specific difference cannot be isolated from the retained original
records. They remain archived; the principal E1 and E2 comparison now
uses the shared corpus and timing framework described in
Section~\ref{subsec:evaluation_scope}. These new measurements do not
retrospectively identify the historical cause.}

\textcolor{blue}{A third limitation concerns the transferability of the E1 checkpoint-context comparison. The observed profile difference is specific to the evaluated inputs and execution conditions. Its magnitude may change with WCL history, storage, caching, or deployment conditions.}

\textcolor{blue}{Fourth, the event-type study uses a fixed four-event
workflow. The separate sensitivity studies in
Section~\ref{subsec:sensitivity_eval} extend the evaluation using
controlled offline evidence fixtures. Their findings apply to the
tested input ranges and verification paths; larger multi-host workflows and arbitrary branching dependency graphs remain subjects for further evaluation.}

\textcolor{blue}{Fifth, the original negative-evidence suite is diagnostic. The separate revision studies add valid-donor substitutions, focused reference-schema variations, and regression conditions under five profiles. They do not constitute exhaustive fuzzing, adaptive adversarial search, or broad parser-robustness validation. Endpoint compromise, key compromise, producer collusion, and semantic errors in agent decisions remain outside the guarantees.}

Finally, the framework evaluates preserved evidence only when such evidence is available. If no preserved bundle is available, the verifier cannot support the corresponding evidence-level claim. Absence of a bundle is not treated as proof of the underlying workflow fact.

\section{Conclusion and Future Work}
\label{sec:conclusion}

\subsection{Conclusion} \label{subsec:conclusion_summary}
\textcolor{blue}{This paper uses preserved evidence bundles to extend message-level authorship and log-inclusion evidence to policy-controlled workflow-event accountability in cross-organization agent messaging. The verifier's \textsc{Accept} decision supports only evidence-level claims justified by the evidence required under the selected policy.}

\textcolor{blue}{The evaluation supports offline evidence verification
using a local prototype and synthetic healthcare workflow cases.
The policy-profile results characterize latency across selected evidence requirements, and the \texttt{submit\_claim} path includes additional capability and prerequisite verification. Under the fully enabled profile, all 17 tested cases in the original E3 diagnostic suite were rejected. Additional functional studies evaluated context and signed-reference binding, while revised E1 and E2 measured the corrected verifier using a shared frozen corpus and timing framework. Together, these studies support the feasibility of policy-relative offline review for the evaluated workloads.}

\subsection{Future Work} \label{subsec:future_work}
Future work will first evaluate the prototype in larger multi-host deployments with concurrent agents and longer checkpoint histories. We will then optimize checkpoint-context and WCL validation to reduce verifier-side latency. 
\textcolor{blue}{Further workflow studies will examine richer branching dependencies and operational multi-hop delegation across organizations beyond the controlled fixtures evaluated here.}
Finally, we will investigate reduced-disclosure evidence views that preserve offline verifiability while exposing less workflow information to external verifiers.






\bibliographystyle{elsarticle-num-names}
\bibliography{Paper2_JISA_1}

@inproceedings{yu_survey_2025,
	title = {A {Survey} on {Agent} {Workflow} – {Status} and {Future}},
	issn = {2769-3554},
	url = {https://ieeexplore.ieee.org/document/11082076},
	doi = {10.1109/ICAIBD64986.2025.11082076},
	urldate = {2026-01-29},
	booktitle = {2025 8th {International} {Conference} on {Artificial} {Intelligence} and {Big} {Data} ({ICAIBD})},
	author = {Yu, Chaojia and Cheng, Zihan and Cui, Hanwen and Gao, Yishuo and Luo, Zexu and Wang, Yijin and Zheng, Hangbin and Zhao, Yong},
	month = may,
	year = {2025},
	note = {{ISSN}: 2769-3554},
	pages = {770--781},
}

@inproceedings{dahlberg_verifiable_2018,
	address = {Cham},
	title = {Verifiable {Light}-{Weight} {Monitoring} for {Certificate} {Transparency} {Logs}},
	isbn = {978-3-030-03638-6},
	doi = {10.1007/978-3-030-03638-6_11},
	language = {en},
	booktitle = {Secure {IT} {Systems}},
	publisher = {Springer International Publishing},
	author = {Dahlberg, Rasmus and Pulls, Tobias},
	editor = {Gruschka, Nils},
	year = {2018},
	pages = {171--183},
}

@article{acharya_agentic_2025,
	title = {Agentic {AI}: {Autonomous} {Intelligence} for {Complex} {Goals}—{A} {Comprehensive} {Survey}},
	volume = {13},
	issn = {2169-3536},
	shorttitle = {Agentic {AI}},
	url = {https://ieeexplore.ieee.org/document/10849561},
	doi = {10.1109/ACCESS.2025.3532853},
	urldate = {2026-01-31},
	journal = {IEEE Access},
	author = {Acharya, Deepak Bhaskar and Kuppan, Karthigeyan and Divya, B.},
	year = {2025},
	pages = {18912--18936},
}

@inproceedings{alshammari_authenticated_2026,
	title = {Authenticated and {Offline}-{Verifiable} {Agent}-to-{Agent} {Messaging} for {LLM} {Agents}},
	url = {https://ieeexplore.ieee.org/document/11393740},
	doi = {10.1109/CCWC67433.2026.11393740},
	urldate = {2026-03-13},
	booktitle = {2026 {IEEE} 16th {Annual} {Computing} and {Communication} {Workshop} and {Conference} ({CCWC})},
	author = {Alshammari, Adil and Assiri, Sareh and Bahsi, Hayretdin},
	month = jan,
	year = {2026},
	pages = {1424--1430},
}

@article{zhao_ai_2026,
	title = {{AI} agent in healthcare: applications, evaluations, and future directions},
	volume = {2},
	copyright = {2026 The Author(s)},
	issn = {3005-1460},
	shorttitle = {{AI} agent in healthcare},
	url = {https://www.nature.com/articles/s44387-026-00076-4},
	doi = {10.1038/s44387-026-00076-4},
	language = {en},
	number = {1},
	urldate = {2026-04-01},
	journal = {npj Artificial Intelligence},
	publisher = {Nature Publishing Group},
	author = {Zhao, Lina and Liu, Shengrui and Xin, Tangsiwei and Tan, Jiawen and Wang, Xiaoran and Li, Yafang and Bian, Zihao and Chen, Yiyang and Kong, Fanyi and Bian, Jinwei and Qian, Chen and Zhang, Zongjiu},
	month = mar,
	year = {2026},
	pages = {31},
}

@article{ferretti_verifiable_2021,
	title = {Verifiable and auditable authorizations for smart industries and industrial {Internet}-of-{Things}},
	volume = {59},
	issn = {2214-2126},
	url = {https://www.sciencedirect.com/science/article/pii/S2214212621000831},
	doi = {10.1016/j.jisa.2021.102848},
	urldate = {2026-04-01},
	journal = {Journal of Information Security and Applications},
	author = {Ferretti, Luca and Longo, Francesco and Merlino, Giovanni and Colajanni, Michele and Puliafito, Antonio and Tapas, Nachiket},
	month = jun,
	year = {2021},
	pages = {102848},
}

@article{dirksen_logpicker_2021,
	title = {{LogPicker}: {Strengthening} {Certificate} {Transparency} {Against} {Covert} {Adversaries}},
	issn = {2299-0984},
	shorttitle = {{LogPicker}},
	url = {https://petsymposium.org/popets/2021/popets-2021-0066.php},
	doi = {10.2478/popets-2021-0066},
	urldate = {2026-04-10},
	journal = {Proceedings on Privacy Enhancing Technologies},
	author = {Dirksen, Alexandra and Klein, David and Michael, Robert and Stehr, Tilman and Rieck, Konrad and Johns, Martin},
	year = {2021},
    number = {4},
    pages = {184--202},
}

@inproceedings{malvai_parakeet_2023,
	address = {San Diego, CA, USA},
	title = {Parakeet: {Practical} {Key} {Transparency} for {End}-to-{End} {Encrypted} {Messaging}},
	isbn = {978-1-891562-83-9},
	shorttitle = {Parakeet},
	url = {https://www.ndss-symposium.org/wp-content/uploads/2023/02/ndss2023_f545_paper.pdf},
	doi = {10.14722/ndss.2023.24545},
	language = {en},
	urldate = {2023-04-10},
	booktitle = {Proceedings 2023 {Network} and {Distributed} {System} {Security} {Symposium}},
	publisher = {Internet Society},
	author = {Malvai, Harjasleen and Kokoris-Kogias, Lefteris and Sonnino, Alberto and Ghosh, Esha and Oztürk, Ercan and Lewi, Kevin and Lawlor, Sean},
	year = {2023},
	pages = {1--18},
}

@article{li_damfsd_2024,
	title = {{DAMFSD}: {A} decentralized authorization model with flexible and secure delegation},
	volume = {27},
	issn = {2542-6605},
	shorttitle = {{DAMFSD}},
	url = {https://www.sciencedirect.com/science/article/pii/S2542660524002580},
	doi = {10.1016/j.iot.2024.101317},
	urldate = {2026-04-10},
	journal = {Internet of Things},
	author = {Li, Minghui and Xue, Jingfeng and Liu, Zhenyan and Suo, Yiran and Lei, Tianwei and Wang, Yong},
	month = oct,
	year = {2024},
	pages = {101317},
}

@article{collaco_role_2026,
	title = {The role of agentic artificial intelligence in healthcare: a scoping review},
	volume = {9},
	copyright = {2026 The Author(s)},
	issn = {2398-6352},
	shorttitle = {The role of agentic artificial intelligence in healthcare},
	url = {https://www.nature.com/articles/s41746-026-02517-5},
	doi = {10.1038/s41746-026-02517-5},
	language = {en},
	number = {1},
	urldate = {2026-04-10},
	journal = {npj Digital Medicine},
	publisher = {Nature Publishing Group},
	author = {Collaco, Bernardo G. and Haider, Syed Ali and Prabha, Srinivasagam and Gomez-Cabello, Cesar A. and Genovese, Ariana and Wood, Nadia G. and Bagaria, Sanjay P. and Gopala, Narayanan and Tao, Cui and Forte, Antonio Jorge},
	month = mar,
	year = {2026},
	pages = {345},
}

@article{lange_we_2025,
	title = {We need accountability in human–{AI} agent relationships},
	volume = {1},
	copyright = {2025 The Author(s)},
	issn = {3005-1460},
	url = {https://www.nature.com/articles/s44387-025-00041-7},
	doi = {10.1038/s44387-025-00041-7},
	language = {en},
	number = {1},
	urldate = {2026-04-10},
	journal = {npj Artificial Intelligence},
	publisher = {Nature Publishing Group},
	author = {Lange, Benjamin and Keeling, Geoff and Manzini, Arianna and McCroskery, Amanda},
	month = nov,
	year = {2025},
	pages = {38},
}

@misc{nguyen_position_2026,
	title = {Position: {Multi}-{Agent} {Algorithmic} {Care} {Systems} {Demand} {Contestability} for {Trustworthy} {AI}},
	shorttitle = {Position},
	url = {http://arxiv.org/abs/2603.20595},
	doi = {10.48550/arXiv.2603.20595},
	language = {en},
	urldate = {2026-04-14},
	publisher = {arXiv},
	author = {Nguyen, Truong Thanh Hung and Fournier, Hélène and Jackson, Piper and Itoh, Makoto and Freeman, Shannon and Richard, Rene and Cao, Hung},
	month = mar,
	year = {2026},
	note = {arXiv:2603.20595 [cs]},
}

@misc{lin_binding_2025,
	title = {Binding {Agent} {ID}: {Unleashing} the {Power} of {AI} {Agents} with accountability and credibility},
	shorttitle = {Binding {Agent} {ID}},
	url = {http://arxiv.org/abs/2512.17538},
	doi = {10.48550/arXiv.2512.17538},
	language = {en},
	urldate = {2026-04-14},
	publisher = {arXiv},
	author = {Lin, Zibin and Zhang, Shengli and Liao, Guofu and Tao, Dacheng and Wang, Taotao},
	month = dec,
	year = {2025},
	note = {arXiv:2512.17538 [cs]},
}

@article{zheng_large_2025,
	title = {Large language model-driven agents in nursing practice: {A} scoping review},
	volume = {12},
	issn = {2352-0132},
	shorttitle = {Large language model-driven agents in nursing practice},
	url = {https://www.sciencedirect.com/science/article/pii/S2352013225001309},
	doi = {10.1016/j.ijnss.2025.10.007},
	number = {6},
	urldate = {2026-04-14},
	journal = {International Journal of Nursing Sciences},
	author = {Zheng, Xinglin and Zou, Huina and Wu, Linjing and Dong, Peihuang and Yuan, Wenhui and Chen, Yuan},
	month = nov,
	year = {2025},
	pages = {532--540},
}

@misc{cloud_security_alliance_more_2026,
	title = {More {Than} {Half} of {Organizations} {Experience} {AI} {Agent} {Scope} {Violations}, {Cloud} {Security} {Alliance} {Study} {Finds}},
	url = {https://cloudsecurityalliance.org/press-releases/2026/04/16/more-than-half-of-organizations-experience-ai-agent-scope-violations-cloud-security-alliance-study-finds},
	urldate = {2026-04-20},
	journal = {Cloud Security Alliance},
	author = {{Cloud Security Alliance}},
	year = {2026},
}

@inproceedings{hu_merkle2_2021,
	title = {Merkle2: {A} {Low}-{Latency} {Transparency} {Log} {System}},
	issn = {2375-1207},
	shorttitle = {Merkle2},
	url = {https://ieeexplore.ieee.org/document/9519459/citations},
	doi = {10.1109/SP40001.2021.00088},
	urldate = {2026-04-24},
	booktitle = {2021 {IEEE} {Symposium} on {Security} and {Privacy} ({SP})},
	author = {Hu, Yuncong and Hooshmand, Kian and Kalidhindi, Harika and Yang, Seung Jin and Popa, Raluca Ada},
	month = may,
	year = {2021},
	note = {{ISSN}: 2375-1207},
	pages = {285--303},
}

@inproceedings{fang_legolog_2025,
	title = {{LegoLog}: {A} configurable transparency log},
	issn = {2995-1356},
	shorttitle = {{LegoLog}},
	url = {https://ieeexplore.ieee.org/document/11129295},
	doi = {10.1109/EuroSP63326.2025.00066},
	urldate = {2026-04-24},
	booktitle = {2025 {IEEE} 10th {European} {Symposium} on {Security} and {Privacy} ({EuroS}\&{P})},
	author = {Fang, Vivian and Dauterman, Emma and Ravoor, Akshay and Dewan, Akshit and Popa, Raluca Ada},
	month = jun,
	year = {2025},
	note = {{ISSN}: 2995-1356},
	pages = {1082--1103},
}

@article{falanji_momep_2026,
	title = {{MoMEP}: {A} formally verified protocol with modifiable signed messages},
	volume = {98},
	issn = {2214-2126},
	shorttitle = {{MoMEP}},
	url = {https://www.sciencedirect.com/science/article/pii/S2214212626000086},
	doi = {10.1016/j.jisa.2026.104378},
	urldate = {2026-04-24},
	journal = {Journal of Information Security and Applications},
	author = {Falanji, Reyhane and Asplund, Mikael and Carlsson, Niklas},
	month = may,
	year = {2026},
	pages = {104378},
}

@article{li_security_2025,
	title = {Security concerns for {Large} {Language} {Models}: {A} survey},
	volume = {95},
	issn = {2214-2126},
	shorttitle = {Security concerns for {Large} {Language} {Models}},
	url = {https://www.sciencedirect.com/science/article/pii/S2214212625003217},
	doi = {10.1016/j.jisa.2025.104284},
	urldate = {2026-04-24},
	journal = {Journal of Information Security and Applications},
	author = {Li, Miles Q. and Fung, Benjamin C. M.},
	month = dec,
	year = {2025},
	pages = {104284},
}

@article{zhu_group-capability-based_2025,
	title = {Group-{Capability}-{Based} {Access} {Control} with {Ring} {Signature}},
	volume = {90},
	issn = {2214-2126},
	url = {https://www.sciencedirect.com/science/article/pii/S2214212625000523},
	doi = {10.1016/j.jisa.2025.104014},
	urldate = {2026-04-24},
	journal = {Journal of Information Security and Applications},
	author = {Zhu, Xiaoying and Zou, Shihong and Xu, Guoai and Xi, Jinwen},
	month = may,
	year = {2025},
	pages = {104014},
}

@article{repetto_cybersecurity_2026,
	title = {Cybersecurity {Digital} {Twins}: {Concept}, blueprint, and challenges for multi-ownership digital service chains},
	volume = {96},
	issn = {2214-2126},
	shorttitle = {Cybersecurity {Digital} {Twins}},
	url = {https://www.sciencedirect.com/science/article/pii/S2214212625003369},
	doi = {10.1016/j.jisa.2025.104299},
	urldate = {2026-04-24},
	journal = {Journal of Information Security and Applications},
	author = {Repetto, M.},
	month = jan,
	year = {2026},
	pages = {104299},
}

@article{tegane_extended_2023,
	title = {An extended {Attribute}-based access control with controlled delegation in {IoT}},
	volume = {76},
	issn = {2214-2126},
	url = {https://www.sciencedirect.com/science/article/pii/S2214212623000571},
	doi = {10.1016/j.jisa.2023.103473},
	urldate = {2026-04-24},
	journal = {Journal of Information Security and Applications},
	author = {Tegane, Saher and Semchedine, Fouzi and Boudries, Abdelmalek},
	month = aug,
	year = {2023},
	pages = {103473},
}

@article{stamnes_karlsen_securing_2026,
	title = {Securing large language models: {A} quantitative assurance framework approach},
	volume = {97},
	issn = {2214-2126},
	shorttitle = {Securing large language models},
	url = {https://www.sciencedirect.com/science/article/pii/S2214212625003874},
	doi = {10.1016/j.jisa.2025.104351},
	urldate = {2026-04-24},
	journal = {Journal of Information Security and Applications},
	author = {Stamnes Karlsen, Sander and Yamin, Muhammad Mudassar and Hashmi, Ehtesham and Katt, Basel and Ullah, Mohib},
	month = mar,
	year = {2026},
	pages = {104351},
}

@article{cheng_s-cred_2024,
	title = {S-{Cred}: {An} accountable anonymous credential scheme with decentralized verification and flexible revocation},
	volume = {82},
	issn = {2214-2126},
	shorttitle = {S-{Cred}},
	url = {https://www.sciencedirect.com/science/article/pii/S2214212624000383},
	doi = {10.1016/j.jisa.2024.103735},
	urldate = {2026-04-24},
	journal = {Journal of Information Security and Applications},
	author = {Cheng, Haotian and Li, Xiaofeng and Zhao, He and Zhou, Tong and Yu, Bin and Sheng, Nianzu},
	month = may,
	year = {2024},
	pages = {103735},
}

@article{hu_towards_2024,
	title = {Towards accountable and privacy-preserving blockchain-based access control for data sharing},
	volume = {85},
	issn = {2214-2126},
	url = {https://www.sciencedirect.com/science/article/pii/S2214212624001686},
	doi = {10.1016/j.jisa.2024.103866},
	urldate = {2026-04-24},
	journal = {Journal of Information Security and Applications},
	author = {Hu, Qiwei and Huang, Chenyu and Zhang, Guoqiang and Cai, Lingyi and Jiang, Tao},
	month = sep,
	year = {2024},
	pages = {103866},
}

@inproceedings{merkle_digital_1988,
	address = {Berlin, Heidelberg},
	title = {A {Digital} {Signature} {Based} on a {Conventional} {Encryption} {Function}},
	isbn = {978-3-540-48184-3},
	doi = {10.1007/3-540-48184-2_32},
	language = {en},
	booktitle = {Advances in {Cryptology} — {CRYPTO} ’87},
	publisher = {Springer},
	author = {Merkle, Ralph C.},
	editor = {Pomerance, Carl},
	year = {1988},
	pages = {369--378},
}

@article{yang_delegation_2024,
	title = {Delegation {Security} {Analysis} in {Workflow} {Systems}},
	volume = {21},
	issn = {1941-0018},
	url = {https://ieeexplore.ieee.org/document/10052704},
	doi = {10.1109/TDSC.2023.3248602},
	number = {1},
	urldate = {2026-04-24},
	journal = {IEEE Transactions on Dependable and Secure Computing},
	author = {Yang, Benyuan and Hu, Hesuan},
	month = jan,
	year = {2024},
	pages = {229--240},
}

@inproceedings{kusters_accountability_2010,
	address = {New York, NY, USA},
	series = {{CCS} '10},
	title = {Accountability: definition and relationship to verifiability},
	isbn = {978-1-4503-0245-6},
	shorttitle = {Accountability},
	url = {https://dl.acm.org/doi/10.1145/1866307.1866366},
	doi = {10.1145/1866307.1866366},
	urldate = {2026-06-01},
	booktitle = {Proceedings of the 17th {ACM} conference on {Computer} and communications security},
	publisher = {Association for Computing Machinery},
	author = {Küsters, Ralf and Truderung, Tomasz and Vogt, Andreas},
	month = oct,
	year = {2010},
	pages = {526--535},
}

@article{yue_glassdb_2023-1,
	title = {{GlassDB}: {An} {Efficient} {Verifiable} {Ledger} {Database} {System} {Through} {Transparency}},
	volume = {16},
	issn = {2150-8097},
	shorttitle = {{GlassDB}},
	url = {https://dl.acm.org/doi/10.14778/3583140.3583152},
	doi = {10.14778/3583140.3583152},
	number = {6},
	urldate = {2026-06-02},
	journal = {Proceedings of the VLDB Endowment},
	author = {Yue, Cong and Dinh, Tien Tuan Anh and Xie, Zhongle and Zhang, Meihui and Chen, Gang and Ooi, Beng Chin and Xiao, Xiaokui},
	month = feb,
	year = {2023},
	pages = {1359--1371},
}

@techreport{laurie_certificate_2021,
	type = {Request for {Comments}},
	title = {Certificate {Transparency} {Version} 2.0},
	url = {https://datatracker.ietf.org/doc/rfc9162},
	doi = {10.17487/RFC9162},
	number = {RFC 9162},
	urldate = {2026-06-02},
	institution = {Internet Engineering Task Force},
	author = {Laurie, Ben and Messeri, Eran and Stradling, Rob},
	month = dec,
	year = {2021},
}

@inproceedings{alshammari_efficient_2025,
	title = {Efficient {Medical} {File} {Protection} via {Adaptive} {CRP}-based {Ephemeral} {Key} {Encapsulation} and {Robust} {Error} {Correction}},
	url = {https://ieeexplore.ieee.org/document/10903918},
	doi = {10.1109/CCWC62904.2025.10903918},
	urldate = {2026-08-27},
	booktitle = {2025 {IEEE} 15th {Annual} {Computing} and {Communication} {Workshop} and {Conference} ({CCWC})},
	author = {Alshammari, Adil and Miandoab, Dina Ghanai and Assiri, Sareh and Cambou, Bertrand and Riggs, Brit},
	month = jan,
	year = {2025},
	pages = {983--988},
}

@article{jakhar_blockchain-based_2024,
	title = {A blockchain-based privacy-preserving and access-control framework for electronic health records management},
	volume = {83},
	copyright = {2024 The Author(s), under exclusive licence to Springer Science+Business Media, LLC, part of Springer Nature},
	issn = {1573-7721},
	url = {https://link.springer.com/article/10.1007/s11042-024-18827-3},
	doi = {10.1007/s11042-024-18827-3},
	language = {En},
	number = {36},
	urldate = {2026-08-27},
	journal = {Multimedia Tools and Applications},
	publisher = {Springer},
	author = {Jakhar, Amit Kumar and Singh, Mrityunjay and Sharma, Rohit and Viriyasitavat, Wattana and Dhiman, Gaurav and Goel, Shubham},
	month = mar,
	year = {2024},
	pages = {84195--84229},
}
\end{document}